\documentclass[11pt]{article}
\usepackage{multirow}
\usepackage[margin=1in,letterpaper]{geometry}
\usepackage{enumerate}
\usepackage{amsmath,amsthm,fancyhdr,mathrsfs,enumerate,latexsym}
\usepackage{graphicx}
\usepackage{amssymb}
\usepackage{multicol}
\usepackage{longtable}
\usepackage{float}
\usepackage{color}
\usepackage{yhmath}

\usepackage[section]{placeins}

\usepackage{tocbibind}
\usepackage{makeidx}
\usepackage{microtype}
\usepackage{natbib}

\usepackage[colorlinks=false,
            pdfpagemode=UseNone,
            urlcolor=blue,
            bookmarks=true,
            backref=page
            ]{hyperref}

\newcommand{\p}{\mathbb{P}}
\newcommand{\E}{\mathrm{E}}
\newcommand{\Var}{\mathrm{Var}}

\newcommand{\brho}{\boldsymbol{\rho}}

\newcommand{\bt}{\pmb{\theta}}

\theoremstyle{plain}
\newtheorem{theorem}{Theorem}[section]
\newtheorem{lemma}[theorem]{Lemma}

\newtheorem{corollary}[theorem]{Corollary}
\theoremstyle{definition}
\newtheorem{example}[theorem]{Example}

\newtheorem{remark}[theorem]{Remark}
\begin{document}
\begin{center}
\Large{\sc Parametric inference for proportional (reverse) hazard rate models with nomination sampling}
\end{center}
\begin{center}
{Mohammad Nourmohammadi, Mohammad Jafari Jozani\footnote{Corresponding author. Email:  m$_{-}$jafari$_{-}$jozani@umanitoba.ca.  Tel: 1 204 272 1563.},  and Brad C.\ Johnson  } \\
\end{center}
\begin{center}
{\it  University of Manitoba, Department of Statistics, Winnipeg, MB, CANADA, R3T 2N2 } 

\end{center}


\begin{abstract}
\noindent Randomized nomination sampling (RNS) is a rank-based sampling technique which  has been shown to be  effective in several nonparametric studies involving  environmental and ecological applications. In this paper, we  investigate parametric inference using RNS design for estimating the unknown vector of parameters ${\bt}$ in the proportional hazard rate  and proportional reverse hazard rate models. We examine both maximum likelihood (ML)  and method of moments (MM) methods  and investigate the relative precision of our proposed RNS-based estimators compared with those based on simple random sampling (SRS).  We introduce four types of  RNS-based data as well as   necessary EM algorithms for the ML estimation, and evaluate the performance of corresponding estimators  in estimating $\bt$.  We show  that  there are always  values of the design parameters on which RNS-based estimators are more efficient than those based on SRS.  Inference based on    imperfect ranking  is also explored and it is shown that the improvement holds even when the ranking is imperfect. Theoretical results are augmented with numerical evaluations and a case study.

\end{abstract}



\section{Introduction}\label{intro}

Randomized nomination sampling (RNS) is a rank-based sampling scheme. Rank-based sampling schemes are data collection techniques which utilize the advantage of additional information available in the population to provide an artificially stratified sample with more structure. Providing more structured sample enables us to direct our attention toward units that  represent the underlying population. 
Let $X$ be an absolutely continuous random variable distributed according to the cumulative distribution function (CDF) $F(x; \bt)$ and the probability density function (PDF) $f(x; \bt)$, where $F$ is known and $\bt=(\theta_1, \ldots, \theta_p)^T\in \Theta\subseteq\mathbb{R}^p$ (p-dimensional Euclidean space), is unknown. Further, let $\{K_i: i\in\mathbb{N}\}$ be a sequence of independent random variables taking values in $\mathbb{N}$ (the natural numbers) with probabilities $\brho=\{(\rho_1,  \rho_2, \ldots): \sum_{i=1}^\infty\rho_i=1\}$ so that $\p(K_i=j)=\rho_j$, $j\in\{1, 2, \ldots\}$. Let $\{Z_i: i\in\mathbb{N}\}$ be a sequence of independent Bernoulli random variables with success probability $\zeta\in[0, 1]$, independent of $K_i$ and $X$.  The RNS design consists of drawing $m$ random sample sets of size $K_i$, $i=1, \ldots, m$, from the population for the purpose of ranking and finally nominating $m$ sampling units (one from each set) for final measurement. The nominee from  the $i$-th set is the largest unit of the set with  probability $\zeta$ or the smallest unit of the set with probability $1-\zeta$. The RNS observation $Y_i$ can be written as
$$Y_i=Z_iX_{K_i:K_i}+(1-Z_i)X_{1:K_i}, ~~~i=1, \ldots, m,$$
where $X_{1:K_i}$ and $X_{K_i:K_i}$ are respectively the smallest and the largest units of the $i^{th}$ set of size $K_i$. 

RNS was introduced by \citet{jafari2012randomized} and applied to the problem of estimating population total in finite population. Later, this sampling scheme was applied in constructing  confidence intervals for quantiles in finite populations  \citep{nourmohammadi2014confidence} and infinite populations \citep{nourmohammadi2014nonparametric}, as well as in constructing  tolerance intervals \citep{nourmohammadi2015distribution}. 
Some well-known examples of RNS are given below:
\begin{itemize}
\item[(1)] The choice of $K_i=1$, $i=1, \ldots, m$, results in the SRS design with observations denoted by $X_i$, $i=1, \ldots, m$. 
 
\item [(2)] The choice of $\zeta=1$ nominates the maximum from each set and results in a maxima nomination sampling (MANS) design (see \citealp{willemain1980estimating};  \citealp{boyles1986estimating}; \citealp{tiwari1988nonparametric};  \citealp{kvam1993estimating}; \citealp{tiwari1989quantile};  \citealp{jafari2010improved}; and \citealp{jafari2011control}).
  
\item [(3)] The choice of $\zeta=0$ nominates the minimum from each set and results in a minima nomination sampling (MINS) design (see \citealp{tiwari1989quantile} and \citealp{wells1990estimating}).
 
\item [(4)] The choice of $\zeta=\frac12$ and $k_i=k$, for a constant $k\in\mathbb{N}$, results in a randomized extreme ranked set sampling (RERSS) design (see \citealp{jafari2012randomized}).
  
\item [(5)] The choice of $Z_{2i-1}=1$ and $Z_{2i}=0$, and $K_i=k$, where $i=1, \ldots, m$, for a constant $k\in\mathbb{N}$ and an even number $m$, results in an extreme ranked set sampling (ERSS) design (see  \citealp{samawi1996estimating}, and  \citealp{ghosh2009unified})
   
\item [(6)] The choice of $Z_{2i-1}=1$ and $Z_{2i}=0$, and $K_i=i$, where  $i=1, \ldots, m$, for  an even number $m$ results in a moving extreme ranked set sampling (MRSS) design (see \citealp{al2001variation}).
\end{itemize}
 Note that $Z_i$, $i=1, \ldots, m$, in (5) and (6) are no longer independent and identically distributed (IID). RNS is a cost-effective method of selecting data in situations in which measuring  the characteristic of interest is difficult, expensive and/or destructive, but a small number of sampling units can be cheaply ranked and the minimum and the maximum observations can be easily identified. Unlike the regular ranked set sampling (RSS), RNS allows for an increase of  the set size without introducing too much ranking  error. Identifying only the extremes, rather than  providing a complete ranking on the units in the set, is more practical, since we need to identify successfully only the first or the last ordered unit. Regarding the randomness of the set size, while the RNS technique allows one to select the sets of fixed size $k$, i.e.   $\p(K_i=k)=1$ for some fixed $k$ and $i=1, \ldots, m$, providing the flexibility of choosing the sizes in random helps to apply the RNS scheme in circumstances where the set size might be random (see \citet{boyles1986estimating} and  \cite{gemayel2010optimal}). Another advantage of allowing the set size to be random is that, when $\rho_1>0$, randomized nomination sample is expected to contain a  simple random sample of size $m\rho_1$  in addition to a collection of extremal order statistics from various set sizes, which contain more information about the underlying population than SRS. RNS also has the flexibility to adjust the proportion of maximums and minimums by choosing an appropriate value for $\zeta\in[0, 1]$ based on the population shape. This reduces the concern in working with extremes when the underlying population is skewed.

The RNS-based statistical inference  may be made under various situations. For example, there might be the case where  $y_i$'s are  the only available information and no further information is provided on either $k_i$ or $z_i$, $i=1, \ldots, m$. There might also  be situations in which the size of sets or the number of maximums (and subsequently the number of minimums), or both are chosen in advance, instead of getting involved in a randomized process. In the situation where $k_i$ and/or $z_i$ are known, the CDF of $Y_i$ can be found by conditioning on $k_i$ and $z_i$, or both. The conditioning argument makes the theoretical investigation more complicated, but it provides  more efficient statistical inference. In this paper, both unconditional and conditional RNS are studied. Four types of RNS data are introduced, corresponding to situations where, for any set $i=1, \ldots, m$, (1) the triplet $(y_i, k_i, z_i)$ are all known, (2) only $(y_i, k_i)$ are known, (3) only $(y_i, z_i)$ are known, or  (4) only $y_i$ are known. These types of RNS data are, respectively, called RNS complete-data,  Type I, Type II, and Type III RNS incomplete-data.

We discuss RNS-based maximum likelihood (ML) and method of moments (MM) estimates of the population parameters when the underlying random variable follows the proportional hazard rate (PHR) or proportional reverse hazard rate (PRHR) model.  
Let $F_0$ be an absolutely continuous probability distribution function with density $f_0$, possibly depending on an unknown vector of parameters $\bt$,  and let $c=\text{sup}\{x\in \mathbb{R}: F_0(x) =0\}$ and $d=\text{inf}\{x\in \mathbb{R}: F_0(x) =1\}$ (where, by convention, we take $\text{inf}~\emptyset =-\text{sup}~\emptyset =\infty$) so that, if $X\sim F_0$, we have $-\infty\leq c\leq X \leq d \leq \infty$ with $F_0(c)=1-F_0(d)=0.$
The family of PHR models (based on $F_0$) is given by the family of distributions
\begin{align}\label{HR-cdf}
F(x; \bt)=1-[1-F_0(x)]^{1/\gamma(\bt)},
\end{align}
where $\gamma(\bt)>0$. Similarly, the family of PRHR models is given by the family of distributions
\begin{align}\label{IHR-cdf}
F(x; \bt)=[F_0(x)]^{1/\gamma(\bt)}.
\end{align}
The hazard rate function $H(x)$ and the reverse hazard rate function $R(x)$ at $x\in (c, d)$ are given respectively by
\begin{align}
H(x)=\frac{1}{\gamma(\bt)}\frac{f_0(x)}{1-F_0(x)}\quad\text{and}\quad R(x)=\frac{1}{\gamma(\bt)}\frac{f_0(x)}{F_0(x)}.
\end{align}
The PHR and PRHR models in (\ref{HR-cdf}) and (\ref{IHR-cdf}) are well-known in lifetime experiments.
 The lifetime distribution of a system and its components are of interest in reliability testing. Statistical analysis of the lifetime of a  system or its components is an important topic in many research areas such as engineering, marketing and biomedical sciences. See, for example,  \citet{lawless2011statistical}, \citet{navarro2008application}, \citet{gupta2007proportional}, \citet{gupta1998modeling} and \citet{helsen1993analyzing}.
The PHR and PRHR models include several well-known lifetime distributions.  
 In the sequel, we are interested in estimating $\gamma(\bt)$ with specific choices of $F_0$. Some examples of  hazard and reverse hazard rate models  are presented in Table \ref{examples}.
\begin{table}[t]
\caption{Some examples of hazard and reverse hazard rate models}
\label{examples}
\centering{
\scalebox{1}{
\begin{tabular}{ cccccc}
\hline\hline
 Distribution &  $F(x; \bt)$ &  Domains & $F_0(x)$ & $\gamma(\bt)$ & Rate Function \\
\hline
\hline
 $X\sim\text{Exp}(\lambda)$ & $1-e^{-x/\lambda}$ & $x\in[0, \infty)$, $\lambda\in(0, \infty)$ & $1-e^{-x}$ & $\lambda$ & $1/\lambda$ \\
$X\sim\text{Par}(\nu=1, \lambda)$ &  $1-x^{-1/\lambda}$ & $x\in [1, \infty)$, $\lambda \in(0, \infty)$ & $1-x^{-1}$ & $\lambda$ & $-1/\lambda x$ \\
$X\sim\text{Bet}(1/\eta, \nu=1)$ & $x^{1/\eta}$ & $ x\in [0, 1]$, $\eta\in(0, \infty)$ & $x$ & $1/\eta$ & $1/\eta x$ \\
$X\sim\text{GExp}(1/\eta, \nu=1)$ & $(1-e^{-x})^{1/\eta}$ & $x\in[0, \infty)$, $\eta\in(0, \infty)$ & $1-e^{-x}$ & $\eta$ & $e^{-x}/\eta (1-e^{-x})$\\
\hline\hline
\end{tabular}
}
}
\end{table}

The remainder of the paper is organized as follows. In Section \ref{MLE},  we investigate the RNS complete-data and provide the PDF and CDF of an RNS observation in the form of complete-data. We also derive the ML estimators of $\gamma(\bt)$ in the PHR and PRHR model when the triplet $(y_i, k_i, z_i)$, $i=1, \ldots, m$, is available.  In Section \ref{MLE-incom}, we present the ML estimation for the parameters based on incomplete RNS data. We provide the PDF and CDF of observations in each RNS incomplete-data  and use the EM algorithm  to obtain the ML estimators of the parameters of interest. In Section \ref{MME}, we derive  the RNS-based MM estimation in the PHR and PRHR models; when the RNS data are from either complete- or incomplete-data scenarios.  In Section \ref{numeric}, we illustrate the numerical  results in detail and compare the performance of the RNS-based estimators  with the  corresponding SRS estimators for the Exponential and Beta distributions. Moreover, in Section \ref{numeric}, the performance of RNS-based ML estimators in a more complicated situation is investigated using a real life dataset on fish mercury contamination  measurement.


\section{ML Estimation in RNS Complete-Data}\label{MLE}

Let $X_1, \ldots, X_m$ be a simple random sample of size $m$ from a continuous distribution with CDF $F(x; \bt)$ and PDF  $f(x; \bt)$. 
If it exists, the SRS-based ML estimator of $\bt$, denoted by  $\hat\bt_{s}$, satisfies the ML equations
\begin{align}\label{SRS-MLE}
\sum_{i=1}^m \frac{f'(X_i; \bt)}{f(X_i; \bt)}=0,
\end{align}
where 
$f'(X; {\bt})=\partial f(X; \bt)/\partial\bt.$
 Let $Y_1, \ldots, Y_m$ be a randomized nomination sample  of size $m$ from $F$. The forms of the CDFs and PDFs of $Y_i$'s, in addition to the RNS-based ML equations, are determined by the available RNS data. 
In this section, we use the RNS complete-data to derive the ML estimator of $\bt$. In the RNS complete-data case, the triplets $(y_i, k_i, z_i)$, $i=1, 2, \ldots,m$,  are known. In other words, one knows that, for $i=1, \ldots, m$, the observed value $Y_i=y_i$ is from a set of size $K_i=k_i$ with the value $Z_i=z_i$ and the rank $r_i=z_ik_i+(1-z_i)$ in the $i$-th set, where $k_i$ and $z_i$ are both known. An RNS observation $Y$ given $K=k$ and $Z=z$, where $k\in\{1, 2, \ldots\}$ and $z\in\{0, 1\}$, has the  CDF $G_c(y | k, z; \bt)$ and the PDF $g_c(y | k, z; \bt)$ as follows
\begin{align*}
G_c(y | k, z; \bt)=\left\{F^k(y; {\bt})\right\}^z\left\{1-\bar F^k(y; {\bt})\right\}^{1-z},
\end{align*}
and
\begin{align*}
 g_c(y | k, z; \bt)=k f(y; {\bt})\left\{F^{k-1}(y; {\bt})\right\}^z\left\{\bar F^{k-1}(y; {\bt})\right\}^{1-z}.
\end{align*}
The log likelihood function based on the RNS complete-data is given by
\begin{align}\label{logLcomplete}
L^{RNS}_c(\bt)
=  \sum_{i=1}^m \left\{\log k_i+ \log f(y_i; {\bt}) +  z_i(k_i-1)\log F(y_i; {\bt})+ 
(1-z_i)(k_i-1)\log\bar F(y_i; {\bt})\right\}.
\end{align}

Upon differentiation of (\ref{logLcomplete}) with respect to ${\bt}=(\theta_1, \ldots, \theta_p)^T$ and equating the result to zero, the (complete) ML estimator of ${\bt}$, denoted by $\hat{\bt}_c$, is obtained from
\begin{align}\label{dLc}
\frac{\partial L_c^{RNS}(\bt)}{\partial{\bt}}&=\sum_{i=1}^m \left\{\frac{f'(y_i; {\bt})}{f(y_i; {\bt})}+ z_i (k_i-1)\frac{F'(y_i; {\bt})}{F(y_i; {\bt})}
+(1-z_i)(k_i-1)\frac{\bar F'(y_i; {\bt})}{\bar F(y_i; {\bt})}\right\}=0,
\end{align}
where $F'(y; {\bt})=\partial F(y; {\bt})/\partial{\bt}$.
Since both $F'(y; {\bt})/F(y; {\bt})$ and $\bar F'(y; {\bt})/\bar F(y; {\bt})$ are involved in the RNS  likelihood,  equation  (\ref{dLc}) is more complicated  to solve for ${\bt}$ than  (\ref{SRS-MLE}), and for most distributions there is no closed form expressions for the ML estimators. Following the idea proposed by \citet{mehrotra1974unbiased}, we consider the modified ML (MML) estimators of parameters.
Depending on the underlying distribution, one may need to replace one or both of the second and third terms on the left-hand side of (\ref{dLc}) by their corresponding expected values. The obtained MML estimator of $\bt$ is denoted by $\hat{\bt}_m$.


\subsection{ML Estimation in the PHR Model}\label{ML-PHR}

Let $X_i$, $i=1, \ldots, m$, be a sequence of IID random variables from the family of PHR models in (\ref{HR-cdf}). The SRS-based ML estimator of $\gamma({\bt})$, denoted by $\widehat{\gamma_s({\bt})}$, can be expressed as
\begin{align}\label{SRS-HR-ML}
\widehat{\gamma_{s}({\bt})}=-\frac 1m\sum_{i=1}^m\log\bar F_0(X_i),
\end{align}
which is an unbiased estimator of $\gamma(\bt)$ with variance $\Var[\widehat{\gamma_s({\bt})}]=\gamma^2({\bt})/m$. 
Under the model  (\ref{HR-cdf}), the PDF of a random variable $Y_i$, $i=1, \ldots, m$, from the RNS complete-data  is 
\begin{align}\label{HR-pdf-y}
g_c(y_i | k_i, z_i; \bt)=k_i \frac{1}{\gamma({\bt})}\frac{f_0(y_i)}{\bar F_0(y_i)}\left([\bar F_0(y_i)]^{\frac{1}{\gamma({\bt})}}\right)^{\alpha_i-1}\left(1-[\bar F_0(y_i)]^{\frac{1}{\gamma({\bt})}}\right)^{\beta_i},
\end{align}
where $\alpha_i=(1-z_i)(k_i-1)+2$ and $\beta_i=z_i(k_i-1)$.
The RNS-based complete ML estimator of $\gamma(\bt)$, denoted by  $\hat \gamma_c({\bt})$, is obtained by solving the equation 
\begin{align}\label{HR-complete}
\sum_{i=1}^m \left\{\gamma({\bt})+\log \bar F_0(Y_i)-\beta_i\frac{[\bar F_0(Y_i)]^{\frac{1}{\gamma({\bt})}}}{1-[\bar F_0(Y_i)]^{\frac{1}{\gamma({\bt})}}}\log \bar F_0(Y_i)+(\alpha_i-2)\log \bar F_0(Y_i)\right\}=0.
\end{align}
Note that if $z_i=0$, for all $i=1, \ldots, m$, one can easily obtain a complete RNS-based ML estimator of $\gamma({\bt})$ as  described in the following lemma. 
\begin{lemma}
Let $X_1, \ldots, X_m$ be  IID random variables from (\ref{HR-cdf}), and suppose $(Y_1, k_1, z_1$=$0)$, \ldots, $(Y_m, k_m, z_m$=$0)$ is the corresponding MINS  sample of size $m$.  The complete ML estimator of $\gamma(\bt)$  is given by
\begin{align}\label{HR-min}
\widehat{\gamma_c({\bt})}=-\frac 1m \sum_{i=1}^mk_i \log \bar F_0(Y_i),
\end{align} 
which is an unbiased estimator of $\gamma(\bt)$ with the variance equal to its SRS-based counterpart, i.e.,
\begin{align*}
\Var[\widehat{\gamma_c({\bt})}]= \frac{\gamma^2({\bt})}{m}.
\end{align*}
\begin{proof}
From  (\ref{HR-complete}), by replacing $z_i=0$ for $i=1, \ldots, m$, the complete ML estimate for $\gamma(\bt)$ in (\ref{HR-min}) is obtained. Noting that in the PHR model we have 
\begin{align}\label{exp-HPR}
\E[\log\bar F_0(Y_i)]&=k_i\text{B}(\alpha_i-1, \beta_i+1)\E(\log W_i)\gamma(\bt),
\end{align}
and
\begin{align}\label{exp^2-HPR}
\E[\log^2\bar F_0(Y_i)]=k_i\text{B}(\alpha_i-1, \beta_i+1)\E(\log^2 W_i)\gamma^2(\bt),
\end{align}
where $W_i\sim\text{Beta}(\alpha_i-1, \beta_i+1)$,  $\text{B}(\alpha, \beta)=\Gamma(\alpha)\Gamma(\beta)/\Gamma(\alpha+\beta)$ and
\begin{align*}
\E[\log W_i]=\mathfrak{D}(\alpha_i-1)-\mathfrak{D}(\alpha_i+\beta_i),
\end{align*}
and
\begin{align*}
\Var[\log W_i]=\mathfrak{D}'(\alpha_i-1)-\mathfrak{D}'(\alpha_i+\beta_i).
\end{align*}
The function $\mathfrak{D}(\cdot)$ is the Digamma function and, for $n\in\mathbb{N}$, $\mathfrak{D}(n)=\sum_{j=1}^{n-1}\frac1j-\tau$, where $\tau\approx 0.57722$ is the Euler-Mascheroni constant. 
The function $\mathfrak{D}'(\cdot)$ is  the Trigamma function, which is defined as
\begin{align*}
\mathfrak{D}'(n)=-\int_0^1\frac{x^{n-1}\log x}{1-x}dx, \quad \text{for}\quad n\in\mathbb{N},
\end{align*}
where $\mathfrak{D}'(n+1)=\mathfrak{D}'(n)-1/n^2$ gives the result. Now, the expected value and the variance of $\widehat{\gamma_c(\bt)}$ follow immediately.
\end{proof}
\end{lemma}
In the general case, to construct the MML estimator of $\gamma({\bt})$ in PHR models, one needs to replace the second term in (\ref{HR-complete})  by its expected value for any $z_i\ne 0$ and $k_i\ne1$. Note that for $z_i=0$ or $k_i=1$ the second term in (\ref{HR-complete}) equals zero.
The RNS-based MML estimator of $\gamma({\bt})$, denoted by $\widehat{\gamma_m({\bt})}$, and the corresponding expected value and variance are given in the following theorem. 
\begin{theorem}\label{HR-com-MLE}
Let $X_1, \ldots, X_m$ be  IID random variables from (\ref{HR-cdf}), and suppose $(Y_1, k_1, z_1)$, \ldots, $(Y_m, k_m, z_m)$ is the corresponding RNS sample of size $m$ with at least one $z_i=1$.  Further, let $U_i\sim\text{Beta}(\alpha_i, \beta_i)$ and $W_i\sim\text{Beta}(\alpha_i-1, \beta_i+1)$,
where $\alpha_i=(1-z_i)(k_i-1)+2$, $\beta_i=z_i(k_i-1)$ and
$$
\mathcal{E}_i=\frac{-1}{\gamma({\bt})}\E\left(\frac{[\bar F_0(Y_i)]^{\frac{1}{\gamma({\bt})}}}{1-[\bar F_0(Y_i)]^{\frac{1}{\gamma({\bt})}}}\log\bar F_0(Y_i)\right)= 
-k_iB(\alpha_i, \beta_i)\E(\log U_i),\quad i=1, \ldots, m.
$$
\begin{enumerate}[(a)]
\item The MML estimator of  $\gamma({\bt})$ based on  RNS complete-data of size $m$ is given by
\begin{align}\label{HR-MMLE}
\widehat{\gamma_m({\bt})}=\frac{-\sum_{i=1}^m(\alpha_i-1)\log\bar F_0(Y_i)}{m+\sum_{i=1}^m\beta_i\mathcal{E}_i},
\end{align}
 
\item The expected value and the variance  of $\hat\gamma_m({\bt})$ are respectively
\begin{align}\label{com-exp}
\E[\widehat{\gamma_m({\bt})}]=\frac{-\sum_{i=1}^m(\alpha_i-1)\E[\log \bar F_0(Y_i)]}{m+\sum_{i=1}^m\beta_i\mathcal{E}_i}\gamma(\bt),
\end{align}
and
\begin{align}\label{com-var}
\Var[\widehat{\gamma_m({\bt})}]=\frac{\sum_{i=1}^m(\alpha_i-1)^2\Var(\log\bar F_0(Y_i))}{(m+\sum_{i=1}^m\beta_i\mathcal{E}_i)^2}\gamma^2(\bt),
\end{align}
where $\E[\log \bar F_0(Y_i)]$ and $\Var[\log\bar F_0(Y_i)]$ are obtained from (\ref{exp-HPR}) and (\ref{exp^2-HPR}).
\end{enumerate}
\end{theorem}

\begin{proof}
For part (a), the value $\mathcal{E}_i$ is derived using the PDF of $Y_i$ in (\ref{HR-pdf-y}). Substituting $\mathcal{E}_i$ for the second term in (\ref{HR-complete}) results in (\ref{HR-MMLE}).  Parts (b) is trivial.
\end{proof}
\begin{remark}
Considering (\ref{com-exp}), 
\begin{align*}
\widetilde{\gamma_m({\bt})}=\left[\frac{m+\sum_{i=1}^m\beta_i\mathcal{E}_i}{-\sum_{i=1}^m(\alpha_i-1)k_iB(\alpha_i-1, \beta_i+1)\E(\log W_i)}\right]\widehat{\gamma_m({\bt})}
\end{align*}
is an unbiased estimator of $\gamma(\bt)$.
\end{remark}
\begin{corollary}
The MML estimator of $\gamma(\bt)$ based on a MANS sample of size $m$ from ($\ref{HR-cdf}$) is given by
\begin{align*}
\widehat{\gamma_m({\bt})}=-\frac{1}{\sum_{i=1}^m\sum_{j=1}^{k_i}\frac1j}\sum_{i=1}^m\log\bar F_0(Y_i),
\end{align*}
which is an  unbiased estimator of  $\gamma(\bt)$ with  variance
\begin{align}\label{var-max-HR}
\Var[\widehat{\gamma_m({\bt})}]=\frac{\sum_{i=1}^m\sum_{j=1}^{k_i}\frac{1}{j^2}}{(\sum_{i=1}^m\sum_{j=1}^{k_i}\frac1j)^2}\gamma^2(\bt),
\end{align}
which is always smaller than its SRS counterpart. 
\end{corollary}

\begin{example}\label{exp-exp}
Let $X_i, i=1, \ldots, m$, be a SRS sample of size $m$ from an exponential distribution with parameter $\lambda$.
 Further, let $(Y_i, k_i, z_i)$ be an RNS sample of size $m$ obtained from the same exponential distribution. Noting that 
\begin{align*}
F(x; \lambda)=1-e^{-x/\lambda} \text{and}\quad f(x; \lambda)=\frac{1}{\lambda}e^{-x/\lambda},
\end{align*}
the SRS-based ML estimator of $\lambda$ is $\hat\lambda_s=\bar X$ with the expected value $\E[\hat\lambda_s]=\lambda$ and the variance $\Var[\hat\lambda_s]=\lambda^2/m$.  
 Assuming $W_i\sim\text{Beta}(\alpha_i-1, \beta_i+1)$, where $\alpha_i=(1-z_i)(k_i-1)+2$ and $\beta_i=z_i(k_i-1)$, we have  
$
\E[Y_i]=-\lambda k_iB(\alpha_i-1, \beta_i+1)\E(\log W_i)$ and $\E(Y_i^2)=\lambda^2k_iB(\alpha_i-1, \beta_i+1)\E(\log^2 W_i).
$
The RNS-based complete ML estimator of $\lambda$, denoted by $\hat\lambda_c$, is obtained from the equation 
\begin{align*}
\sum_{i=1}^m \left\{\lambda-(\alpha_i-1)Y_i+\beta_i\frac{e^{-Y_i/\lambda}}{1-e^{-Y_i/\lambda}}Y_i\right\}=0.
\end{align*}
The MINS complete-data ML estimator of $\lambda$ is given by
\begin{align*}
\hat\lambda_c=\frac{\sum_{i=1}^m k_iY_i}{m},
\end{align*}
which is unbiased for $\lambda$ with the variance  $\Var[\hat\lambda_c]=\lambda^2/m$. 
Also, the MML estimator of $\lambda$ based on  MANS complete-data is
\begin{align}\label{exp-mans}
\hat\lambda_m=\frac{\sum_{i=1}^m Y_i}{\sum_{i=1}^m\sum_{j=1}^{k_i}\frac1j}.
\end{align}
Noting that 
$
\E[Y_i]= \lambda\sum_{j=1}^{k_i}\frac{1}{j}
$
and
$
\Var[Y_i]=\lambda^2\sum_{j=1}^{k_i}\frac{1}{j^2},
$
 the MML estimator $\hat\lambda_m$  in (\ref{exp-mans}) is unbiased with  variance
\begin{align*}
\Var[\hat\lambda_m]=\frac{\sum_{i=1}^m\sum_{j=1}^{k_i}\frac{1}{j^2}}{(\sum_{i=1}^m\sum_{j=1}^{k_i}\frac1j)^2} \lambda^2.
\end{align*} 
\end{example}
\begin{remark}
It can be shown that the above mentioned results hold with  minor modifications for  the family of PRHR model in (\ref{IHR-cdf}). 
For example,
the RNS-based complete ML estimator of the parameter $\gamma({\bt})$, is obtained from
\begin{align}\label{IHR-complete}
\sum_{i=1}^m\left\{\gamma({\bt})+\log F_0(Y_i)+\beta_i\log F_0(Y_i)-(\alpha_i-2)\left(\frac{[F_0(Y_i)]^{\frac{1}{\gamma({\bt})}}}{1-[F_0(Y_i)]^{\frac{1}{\gamma({\bt})}}}(\log F_0(Y_i))\right)\right\}=0.
\end{align}
Also, the RNS-based complete ML estimator of $\gamma(\bt)$ in (\ref{IHR-cdf}) using  MANS sample is given by
\begin{align*}
\widehat{\gamma_c({\bt})}=-\frac{1}{m}\sum_{i=1}^mk_i \log F_0(Y_i),
\end{align*} 
which is an unbiased estimator of $\gamma(\bt)$ with  variance equal to its SRS-based counterpart, i.e., $\lambda^2/m$.
Note that the ML estimator of $\gamma(\bt)$ in  PRHR models based on SRS is given by
$$
\widehat{\gamma_s(\bt)}=-\frac{1}{m}\sum_{i=1}^m \log F_0(X_i).
$$
The MML estimator of $\gamma(\bt)$ in the case where at least one $z_i=0$ is given by
\begin{align*}
\widehat{\gamma_m({\bt})}=\frac{-\sum_{i=1}^m(\beta_i+1)\log F_0(Y_i)}{m+\sum_{i=1}^m(\alpha_i-2)\mathcal{E}^*_i},
\end{align*}
where
$\mathcal{E}^*_i=  -k_iB(\beta_i+2, \alpha_i-2)\E(\log U^*_i),$
and $U_i^*\sim\text{Beta}(\beta_i+2, \alpha_i-2)$.
The expected value  and the variance of $\widehat{\gamma_m({\bt})}$ are, respectively, 
\begin{align}\label{E-PRHR}
\E[\widehat{\gamma(\bt)}]=\frac{-\sum_{i=1}^m(\beta_i+1)k_i\text{B}(\beta_i+1, \alpha_i-1)\E[\log E_i^*]}{m+\sum_{i=1}^m(\alpha_i-2)\mathcal{E}_i^*}\gamma(\bt)
\end{align}
and
\begin{align}\label{V-PRHR}
\Var[\widehat{\gamma(\bt)}]=\frac{\sum_{i=1}^m(\beta_i+1)^2\Var[\log\bar F_0(Y_i)]}{(m+\sum_{i=1}^m(\alpha_i-2)\mathcal{E}_i^*)^2}\gamma^2(\bt),
\end{align}
where, assuming $W^*_i\sim\text{Beta}(\beta_i+1, \alpha_i-1)$, we have
\begin{align}\label{exp-PRH}
\E[\log F_0(Y_i)]&=\gamma({\bt})k_iB(\beta_i+1, \alpha_i-1)\E(\log W^*_i), 
\end{align}
and
\begin{align}\label{exp^2-PRH}
\E[\log^2 F_0(Y_i)]&=\gamma^2({\bt})k_iB(\beta_i+1, \alpha_i-1)\E(\log^2 W^*_i).
\end{align}

In  the MINS design,  the MML estimator of $\gamma({\bt})$ is given by
\begin{align*}
\widehat{\gamma_m({\bt})}=\frac{-\sum_{i=1}^m\log F_0(Y_i)}{\sum_{i=1}^m\sum_{j=1}^{k_i}\frac1j}.
\end{align*}
One can easily show that $\widehat{\gamma_m(\bt)}$ is an unbiased estimator of $\gamma(\bt)$ with  the variance given by  (\ref{var-max-HR}).
\end{remark}
 
 \begin{example}\label{beta-beta}
 Let $X_1, \ldots, X_m$ be an SRS sample of size $m$ from a Beta($\frac1\eta, 1$) distribution with  the corresponding  CDF and PDF as follow
 \begin{align}\label{beta-cdf}
 F(x; {\bt})=x^\frac1\eta \quad\text{and}\quad f(x; {\eta})=\frac1\eta x^{\frac1\eta-1}.
 \end{align}
 Further, let $Y_i$, $i=1, \ldots, m$, denote an RNS sample from  (\ref{beta-cdf}). 
 The RNS-based complete MLE, $\hat\eta_c$, is obtained from the equation
 \begin{align*}
 \sum_{i=1}^m\left\{\eta+\log y_i+z_i(k_i-1)\log Y_i+(1-z_i)(k_i-1)\frac{Y_i^{\frac{1}{\eta}}}{1-Y_i^{\frac{1}{\eta}}}(-\log Y_i) \right\}=0.
 \end{align*}
 For a MANS sample, the unbiased RNS-based ML estimate for $\eta$ using the complete-data is given by
 \begin{align*}
 \hat\eta_c=-\frac{\sum_{i=1}^m k_i\log Y_i}{m},
 \end{align*}
 with $\Var(\hat\eta_c)=\eta^2/m$.
 Also, based on the MINS design, the parameters $\alpha_i$ and $\beta_i$ are, respectively, 1 and $k_i$. It is seen that the unbiased  MML estimator of $\eta$ based on RNS complete-data is given by
 \begin{align*}
 \hat\eta_m=\frac{-\sum_{i=1}^m\log Y_i}{\sum_{i=1}^m\sum_{j=1}^{k_i}\frac1j},
 \end{align*}
 with the variance provided in (\ref{var-max-HR}).
 \end{example}
 
 \section{ML Estimation in RNS Incomplete-Data}\label{MLE-incom}
 
 In RNS complete-data the triplet $(y_i, k_i, z_i)$ are all assumed to be observed. In practice, for $i=1, \ldots, m$, one of $k_i$ or $z_i$, or both, may be unknown. In this section, we investigate the CDF  and PDF of the RNS random variable $Y_i$, the ML equations and corresponding EM algorithms associated with each type of RNS incomplete-data.
 
\begin{enumerate}[(a)] 
\item {\bf Type I RNS incomplete-data:} 
Here,  we assume that $(y_i, k_i)$, for $i=1, \ldots, m$, are known. In other words, the $i$-th observed unit is from a set of size $k_i$. This unit may be the maximum or minimum with probability $\zeta$ or $1-\zeta$, respectively. The CDF $G_1(y | k; \bt)$ and PDF $g_1(y | k; \bt)$ for the observed values $Y=y$ given $K=k$ are, respectively, as follow
\begin{align*}
G_1(y| k; \bt)&=
\zeta F^k(y; {\bt})+(1-\zeta)(1-\bar F^k(y; {\bt})),~ \text{and}\\
g_1(y | k; \bt)
&=kf(y; {\bt})\left\{\zeta F^{k-1}(y; {\bt})+ (1-\zeta)\bar F^{k-1}(y; {\bt})\right\}.
\end{align*}
The log likelihood function  based on ${\boldsymbol y}=(y_1, y_2, \ldots, y_m)$ and ${\boldsymbol k}=(k_1, k_2, \ldots, k_m)$ is given by
\begin{align}\label{T1logL}
L_1^{RNS}&=\sum_{i=1}^m\log k_i+\sum_{i=1}^m\log f(y_i; {\bt})
+\sum_{i=1}^m\log\left\{\zeta F^{k_i-1}(y_i; {\bt})+(1-\zeta)\bar F^{k_i-1}(y_i; {\bt})\right\}.
\end{align}
Upon differentiating (\ref{T1logL}) with respect to ${\bt}$ and equating the result to zero, we have

\begin{align*}
\frac{\partial L_1^{RNS}}{\partial{\bt}}&=\sum_{i=1}^m\frac{f'(y_i; {\bt})}{f(y_i; {\bt})} \nonumber\\
&+\sum_{i=1}^m \left\{\frac{\zeta(k_i-1)F^{k_i-2}(Y_i; {\bt})F'(Y_i; {\bt})+(1-\zeta)(k_i-1)\bar F^{k_i-2}(Y_i; {\bt})\bar F'(Y_i; {\bt})}{\zeta F^{k_i-1}(y_i; {\bt})+(1-\zeta)\bar F^{k_i-1}(y_i; {\bt})}\right\}\\
&=0, 
\end{align*}
which does not yield explicit solutions for $\hat\bt$. In order to find the ML estimators, we use the EM algorithm.
In order to pose this problem in the incomplete-data Type I context, we introduce the unobservable data ${\bf z}=(z_1, z_2, \ldots, z_m)$, where $z_i=1$ or $z_i=0$ 
according to whether the selected unit in the 
$i$-{th} set is the maximum or the minimum, respectively.  
We find the ML estimator of ${\bt}$ by adding the unobservable data to the problem via working with the current conditional expectation of the complete-data log likelihood (\ref{logLcomplete}) given the observed data and proceed as follow. 
 Let ${\bt}^{(t)}$ be  the value specified  for ${\bt}$ in the $t$-{th} iteration. Then on the $(t+1)$-{th} iteration,  the conditional expectation of $L_c^{RNS}$ given $\boldsymbol y$ and $\boldsymbol k$ using ${\bt}^{(t)}$, i.e., $\E_{{\bt}^{(t)}}\left[L_c^{RNS} | \boldsymbol y, \boldsymbol k\right]$ is computed. This step is called the E-step.
As $L_c^{RNS}$ is a linear function of the unobservable data $\boldsymbol z$, the E-step is performed  simply by replacing $z_i$ by their current conditional expectations given the observed data $\boldsymbol y$ and $\boldsymbol k$.
Therefore, for a known parameter $\zeta$, we have
\begin{align*}
\E_{{\bt}^{(t)}}[Z_i  |  {\boldsymbol y}, {\boldsymbol k}]= \p_{{\bt}^{(t)}}(Z_i=1 | y_i, k_i)=z_i^{(t)},
\end{align*}
where
\begin{align*}
z_i^{(t)}=\frac{g_c(y_{i}| k_i, z_i=1, \bt=\bt^{(t)})}{g_1(y_{i} | k_i, \bt=\bt^{(t)})}=\frac{\zeta F^{k_i-1}(y_{i}; \bt^{(t)})}{\zeta F^{k_i-1}(y_{i}; \bt^{(t)})+(1-\zeta)\bar F^{k_i-1}(y_{i}; \bt^{(t)})}.
\end{align*}
The next step  on the $(t+1)$-{th} iteration, which is called the M-step, requires replacing $z_i$'s by $z_i^{(t)}$'s  in (\ref{logLcomplete}) to obtain ${\bt}^{(t+1)}$ by maximizing $L_c^{RNS}(\bt)$. We keep alternating between  $z_i^{(t)}$ and $\bt^{(t)}$ until $\bt^{(t)}$  converges to a fixed point.

When the parameter $\zeta$ is unknown, the procedure may be started with the initial value of $\zeta^{(0)}\approx 1$ (in PHR model) or $\zeta^{(0)}\approx 0$ (in PRHR model), and continued by updating $\zeta$ using $\zeta^{(t+1)}=\frac1m\sum_{i=1}^mz_i^{(t)}$.

\item {\bf Type II RNS incomplete-data:} 
Here, we consider the case where $(y_i, z_i)$  are known, but the set size $k_i$ is unknown. In other words, we observed the value $y_i$ and we know if the observed unit is the maximum or the minimum unit of the set, but the set size is unknown. The CDF $G_2(y | z; \bt)$ and PDF $g_2(y | z; \bt)$ for the observed value $Y=y$ given $Z=z$ are, respectively, as follow
\begin{align}\label{T2-cdf}
G_2(y | z; \bt)&=\sum_{k=1}^\infty \p(K=k)G_c(y | k, z; \bt)=\sum_{k=1}^\infty\rho_k\left\{F^k(y; \bt)\right\}^z\left\{(1-\bar F^k(y; \bt)\right\}^{1-z},
\end{align}
and
\begin{align*}
g_2(y | z; \bt)
= f(y, \bt)\sum_{k=1}^\infty k \rho_k \left\{F^{k-1}(y; \bt)\right\}^z\left\{\bar F^{k-1}(y; \bt)\right\}^{1-z}.
\end{align*}
From (\ref{T2-cdf}), the log likelihood function for ${\bt}$ obtained from the observed data ${\boldsymbol y}=(y_1, y_2, \ldots, y_m)$ and ${\boldsymbol z}=(z_1, z_2, \ldots, z_m)$ is expressed as
\begin{align}\label{L2}
L_2^{RNS}&= \sum_{i=1}^m\log f(y_i, \bt) 
+\sum_{i=1}^m\log\left[\sum_{k_i=1}^\infty k_i\rho_{k_i}\left\{F^{k_i-1}(y_i, \bt)\right\}^{z_i}\left\{\bar F^{k_i-1}(y_i, \bt)\right\}^{1-z_i}\right].
\end{align}
Upon equating the derivatives of (\ref{L2}) with respect to ${\bt}$ to zero, we have
\begin{align*}
\frac{\partial L_2^{RNS}}{\partial{\bt}}&=\sum_{i=1}^m\frac{f'(y_i, \bt)}{f(y_i, \bt)}\nonumber\\
&+\sum_{i=1}^m (-1)^{1-z_i}f(y_i, \bt)\left\{\frac{\sum_{k_i=1}^\infty k_i(k_i-1)\rho_{k_i}\left\{F^{k_i-2}(y_i,\bt)\right\}^{z_i}\left\{\bar F^{k_i-2}(y_i, \bt)\right\}^{1-z_i}}{\sum_{k_i=1}^\infty k_i\rho_{k_i}\left\{F^{k_i-1}(y_i, \bt)\right\}^{z_i}\left\{\bar F^{k_i-1}(y_i, \bt)\right\}^{1-z_i}}\right\}\nonumber\\
&=0,
\end{align*}
which apparently do not yield explicit solutions for the incomplete-data MLE of ${\bt}$. Since the vector ${\boldsymbol k}=(k_1, k_2, \ldots, k_m)$ is unobservable, we are unable to estimate ${\bt}$ by the maximizing (\ref{logLcomplete}). So we again use the EM algorithm. In the E-step, we substitute the unobservable data in (\ref{logLcomplete}) by averaging the complete-data log likelihood over its conditional distribution given the observed ${\boldsymbol y}$ and ${\boldsymbol z}$. As $L_c^{RNS}$ is a linear function of the unobservable data $\boldsymbol k$, the E-step is performed simply by replacing $k_i$ by their current conditional expectations given the observed data $\boldsymbol y$ and $\boldsymbol z$.
Therefore, for a known parameter $\brho$,
\begin{align*}
\E_{{\bt}^{(t)}}[K_i  |  {\boldsymbol y}, {\boldsymbol z}]= \sum_{k=1}^\infty k\p_{{\bt}^{(t)}}(K_i=k | {\boldsymbol y}, {\boldsymbol z})= k_i^{(t)},
\end{align*}
where 
\begin{align*}
 k_i^{(t)}&= \sum_{k=1}^\infty k\rho_k\frac{g_c(y_{i, t} | z_i, k_i=k, \bt=\bt^{(t)})}{g_2(y_{i, t} | z_i, \bt=\bt^{(t)})}\\
 &=\frac{\sum_{k=1}^\infty k^2\rho_k \left\{F^{k_i-1}(y_{i}; \bt^{(t)})\right\}^{z_i}\left\{\bar F^{k_i-1}(y_{i}; \bt^{(t)})\right\}^{1-z_i}}{\sum_{k=1}^\infty k\rho_k\left\{F^{k_i-1}(y_{i}; \bt^{(t)})\right\}^{z_i}\left\{\bar F^{k_i-1}(y_{i}; \bt^{(t)})\right\}^{1-z_i}}.
\end{align*}

The M-step on the $(t+1)$-{th} iteration requires replacing $k_i$'s by $k_i^{(t)}$'s  in (\ref{logLcomplete}) to obtain ${\bt}^{(t+1)}$ by maximizing $L_c^{RNS}(\bt)$. We keep alternating between  $
k_i^{(t)}$ and $\bt^{(t)}$ until $\bt^{(t)}$  converges to a fixed point.

When the parameter $\brho$ is unknown, the procedure can be started with the initial value $\rho_1^{(0)}=(1/n, \ldots, 1/n)$, where the length of $\rho_1^{(0)}$ is a relatively large but arbitrary $n$, and continued  by calculating $\rho_{k_i}^{(t+1)}$ using  the frequencies of the $k_i^{(t)}$ over $m$.

\item {\bf Type III RNS incomplete-data:} Here, we study the case  where only $y_i$ is observed and no more information on the set size and the rank of the selected unit is  available. The CDF $G_3(y | \bt)$ and PDF $g_3(y | \bt)$ for the observed value $Y=y$  are given, respectively, by
\begin{align*}
G_3(y ; \bt)
&= \zeta \sum_{k=1}^\infty \rho_kF^k(y; \bt)+(1-\zeta)\sum_{k=1}^\infty \rho_k(1-\bar F^k(y; \bt)), 
\end{align*}
and
\begin{align*}
g_3(y ; \bt)
&=f(y; \bt) \left\{\zeta \sum_{k=1}^\infty k \rho_k F^{k-1}(y; \bt) +(1-\zeta)\sum_{k=1}^\infty k \rho_k\bar F^{k-1}(y; \bt)\right\}.
\end{align*}
The log likelihood function for ${\bt}$ formed on the basis of ${\bf y}$ is given by
\begin{align}\label{L3}
L_3^{RNS}&=\sum_{i=1}^m \log f(y_i,\bt)
+\sum_{i=1}^m \log\left\{\zeta \sum_{k=1}^\infty k \rho_k F^{k-1}(y; \bt) +(1-\zeta)\sum_{k=1}^\infty k \rho_k\bar F^{k-1}(y; \bt)\right\}.
\end{align}
Upon equating the derivatives of  (\ref{L3}) with respect to $\bt$ to zero, the following results are obtained: 
\begin{align}\label{L3der}
\frac{\partial L_3^{RNS}}{\partial{\bt}}&=\sum_{i=1}^m \frac{f'(y_i; \bt)}{f(y_i;  \bt)}\nonumber\\
&+\sum_{i=1}^m f(y_i;  \bt)\left\{\frac{\zeta\sum_{k=1}^\infty k(k-1)\rho_k F^{k-2}(y_i;  \bt)-(1-\zeta)\sum_{k=1}^\infty k(k-1)\rho_k\bar F^{k-2}(y_i; \bt)}{\zeta\sum_{k=1}^\infty k\rho_k F^{k-1}(y_i;  \bt)+(1-\zeta)\sum_{k=1}^\infty k\rho_k\bar F^{k-1}(y_i; \bt)}\right\}\nonumber \\
&=0.
\end{align}
Similar to Type I and Type II incomplete-data, no explicit ML estimator for the parameter  ${\bt}$ can be obtained from (\ref{L3der}). In this type of RNS incomplete-data, two unobservable data sets in the form of ${\boldsymbol z}=(z_1, z_2, \ldots, z_m)$ and ${\boldsymbol k}=(k_1, k_2, \ldots, k_m)$ are introduced. In order to 
perform the EM algorithm assuming $\brho$ and $\zeta$ are known, we first calculate
 \begin{align*}
\E_{{\bt}^{(t)}}[K_i  |  {\boldsymbol y}]= \sum_{k=1}^\infty k\p_{\bt^{(t)}}(K_i=k | {\boldsymbol y})=k_i^{(t)},
\end{align*}
where
\begin{align*}
k_i^{(t)}
=\frac{\sum_{k=1}^\infty k^2\rho_k \{\zeta F^{k-1}(y_{i}; \bt^{(t)})+(1-\zeta)\bar F^{k-1}(y_{i}; \bt^{(t)})\}}{\zeta \sum_{k=1}^\infty k \rho_k F^{k-1}(y_{i}; \bt^{(t)}) +(1-\zeta)\sum_{k=1}^\infty k \rho_k\bar F^{k-1}(y_{i}; \bt^{(t)})}.
\end{align*}
Then, we obtain
\begin{align*}
\E_{\bt^{(t)}}[Z_i|{\boldsymbol y}]=\p_{{\bt}^{(t)}}(Z_i=1 | {\boldsymbol y})=z_i^{(t)},
\end{align*}
where 
\begin{align*}
z_i^{(t)}
=\frac{\zeta F^{k^{(t)}_i-1}(y_{i}; \bt^{(t)})}{\zeta F^{k^{(t)}_i-1}(y_{i}; \bt^{(t)})+(1-\zeta)\bar F^{k^{(t)}_i-1}(y_{i}; \bt^{(t)})}.
\end{align*}

 The M-step on the $(t+1)$-{th} iteration requires replacing $z_i$'s by $z_i^{(t)}$'s and $k_i$'s by $k_i^{(t)}$'s  in (\ref{logLcomplete}) to obtain ${\bt}^{(t+1)}$ by maximizing $L_c^{RNS}(\bt)$. We keep alternating between  $z_i^{(t)}$, $k_i^{(t)}$ and $\bt^{(t)}$ until $\bt^{(t)}$  converges to a fixed point.
 When the parameters $\zeta$ and $\brho$ are unknown, similar procedures proposed in Type I incomplete-data (for $\zeta$) and in Type II incomplete-data (for $\brho$) are used.
 \end{enumerate}

\section{RNS-Based MM Estimators}\label{MME}

 Finding the ML estimators of $\bt$ for complete-data case requires finding the roots of the nonlinear equations (\ref{HR-complete}) and (\ref{IHR-complete}), which are cumbersome and computationally expensive.  When the available data is incomplete, the iterative EM algorithm for calculating the ML estimator of $\bt$ is not easy-to-use. 
In this section, we briefly study the MM estimation  based on RNS data for $\gamma(\bt)$ in PHR and PRHR models. The SRS-based  MM estimate of $\gamma(\bt)$, denoted by $\widehat{\gamma_{sm}(\bt)}$ is equal to the SRS-based ML estimate, $\widehat{\gamma_s(\bt)}$, obtained from (\ref{SRS-HR-ML}). In PHR, considering the random variable $\log\bar H(X)$, the MM estimator of $\gamma(\bt)$ can be obtained  by equating the first  moment of the population  to the sample moment as follow
\begin{align*}
\widehat{\gamma_{sm}(\bt)}=-\frac1m\sum_{i=1}^m \log\bar F_0(X_i).
\end{align*}
Similarly, in PRHR model, the MM estimator of $\gamma(\bt)$ is expressed as
\begin{align*}
\widehat{\gamma_{sm}(\bt)}=-\frac1m\sum_{i=1}^m \log F_0(X_i).
\end{align*}
Now, we present the RNS-based complete- and incomplete-data MM estimators of $\gamma(\bt)$ in both PHR and PRHR models.
\begin{theorem}\label{MMestim}
Let $Y_1, \ldots, Y_m$ be  an RNS sample of size $m$ obtained from  a continuous CDF of the family of PHR model or PRHR model.  Further, let $W_i\sim\text{Beta}(\alpha_i-1, \beta_i+1)$ and $W^*_i\sim\text{Beta}(\beta_i+1, \alpha_i-1)$, where $\alpha_i=(1-z_i)(k_i-1)+2$ and $\beta_i=z_i(k_i-1)$. Then, the unbiased MM estimators of $\gamma(\bt)$ in PHR and PRHR models are, respectively, obtained as
\begin{align*}
\widehat{\gamma_m(\bt)}=-\frac1m\sum_{i=1}^m\frac{1}{R_{i}}\log\bar F_0(Y_i)\quad\text{and}\quad\widehat{\gamma_m(\bt)}=-\frac1m\sum_{i=1}^m\frac{1}{R_{i}}\log F_0(Y_i),
\end{align*}
where the value of $R_i$ depends on the RNS data type and the underlying model as presented in Table \ref{MM}.
\begin{table}[t]
\caption{The value of $R_i$ introduced in Theorem \ref{MMestim} for RNS complete-data and Type I, Type II and Type II incomplete-data}
\label{MM}
\centering{
{
\begin{tabular}{ ccc }
  \hline\hline
  RNS Data & \multicolumn{2}{c}{$R_i$} \\
  \cline{2-3}
   & \text{PHR} & \text{PRHR}\\
  \hline
  \hline
  \text{Complete} & $-k_iB(\alpha_i-1, \beta_i+1)\E[\log W_i]$ & $-k_iB(\beta_i+1, \alpha_i-1)\E[\log W^*_i]$  \\
  \text{Type I} & $\zeta\sum_{j=1}^{k_i}\frac1j+(1-\zeta)\frac{1}{k_i}$   &  $(1-\zeta)(\sum_{j=1}^{k_i}\frac1j)+\zeta\frac{1}{k_i}$ \\
  \text{Type II} &  $-\sum_{k_i=1}^\infty\rho_{k_i} k_i B(\alpha_i-1, \beta_i+1)\E[\log W_i]$  &  $-\sum_{k_i=1}^\infty\rho_{k_i} k_i B(\beta_i+1, \alpha_i-1)\E[\log W^*_i]$ \\
  \text{Type III} & $\sum_{k_i=1}^\infty\rho_{k_i} \left\{\zeta\sum_{j=1}^{k_{i}}\frac1j+ (1-\zeta)\frac{1}{k_i}\right\}$ & $\sum_{k_i=1}^\infty\rho_{k_i} \left\{(1-\zeta)\sum_{j=1}^{k_{i}}\frac1j+ \zeta\frac{1}{k_i}\right\}$ \\
  \hline\hline
\end{tabular}
}
}
\end{table}
  \end{theorem}
Note that for the RNS complete-data, the variance of  $\widehat{\gamma_m(\bt)}$ in   both PHR and  PRHR models provided in Theorem \ref{MMestim} are derived using (\ref{exp-HPR}), (\ref{exp^2-HPR}), (\ref{exp-PRH}), and (\ref{exp^2-PRH}).


\section{Numerical Studies}\label{numeric}

In this section, we perform numerical studies to compare the performance of the proposed  RNS methods with SRS in estimating some parameters. First, we perform some simulations to examine the performance of RNS compared with SRS under different scenarios of available information about the observations, set sizes, and rank of observations. Then, in a case study we evaluate the precision of the RNS design over the SRS design  in a more complicated scenario in both perfect and imperfect settings.

\subsection{Simulations}\label{simulation}

We first discuss the reduction  in the mean square error (MSE) of the ML estimators in the RNS complete-data in the PHR and PRHR models using the relative precision. The relative precision is defined as the ratio of the RNS-based MSE over the SRS-based MSE such that values less than one are desired.   For the incomplete-data settings, the performance of   MLE's of the population parameters in two distributions are examined; the parameter $\lambda$ in the exponential distribution introduced in Example \ref{exp-exp} and the parameter $\eta$ in the beta distribution introduced in Example \ref{beta-beta}. 
Note that the expected value and the variance of the RNS complete-data in the PHR and PRHR models presented in (\ref{com-exp}), (\ref{com-var}), (\ref{E-PRHR}), and (\ref{V-PRHR}) do not depend on the observed data and are only functions of $K$ and $Z$. In addition, we investigate the role of the RNS design parameters in improving the performance of the RNS-based estimators compared with their SRS counterparts.

In Figure \ref{msefixedk-PHR}, we provide the MSE of $\widehat{\gamma_c(\bt)}$, the estimator of $\gamma(\bt)$  in the RNS setting when the data is complete,  over the MSE of $\hat\gamma_s(\bt)$, the estimator of $\gamma(\bt)$ in the SRS setting, for the PHR (left panel) and PRHR (right panel) models. The relative precision is calculated  for four RNS designs with fixed set sizes as $K=2, 3, 4, 5$ and the proportion of maximums varies in $p=0(0.1)1$. From the results we can compare RNS complete-data with different fixed set sizes and proportion of maximums among themselves and with SRS in terms of the performance of estimators of  $\gamma(\bt)$. For example, in the left panel of Figure \ref{msefixedk-PHR}, which shows the relative precision for the PHR models, it is seen that any RNS design, even with $K=2$ and proportion of maximums $p=0.1$, outperforms SRS. Increasing the set size and the proportion of maximums improves the performance of the RNS complete-data. The best performance pertains to MANS with $K=5$. In the right panel of Figure \ref{msefixedk-PHR}, which shows the relative precision for the PRHR models, similar results are obtained except  that the best performance pertains to MINS with $K=5$.

\begin{figure}
 \centering  \vspace{-1cm}              
  \includegraphics[width=6in, height=2.5in]{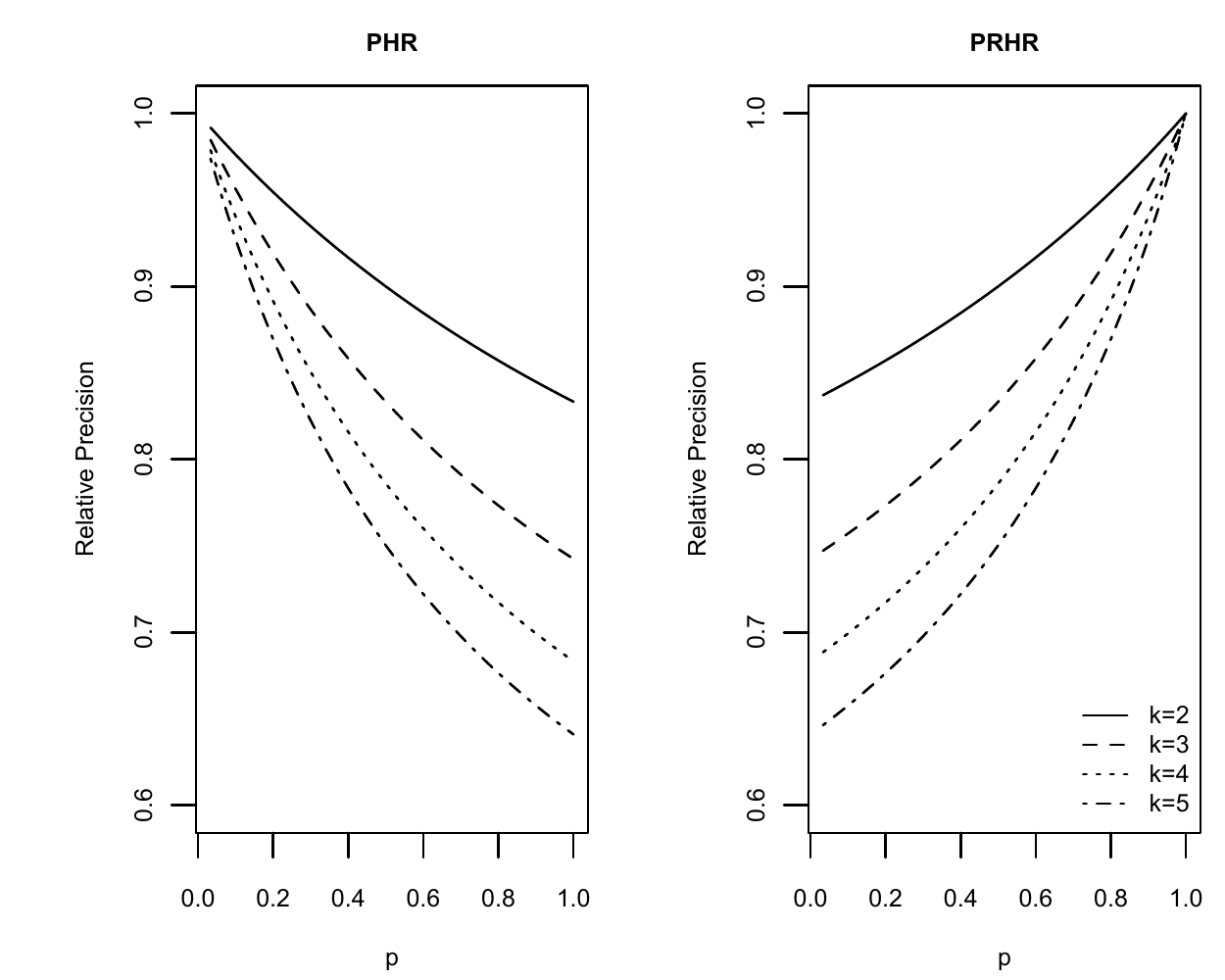} 
  \caption{{The relative precision of the ML estimators of $\gamma(\bt)$ based on the RNS complete-data over  their SRS counterparts in the PHR (left panel) and PRHR (right panel) models when $K=2, 3, 4, 5$ and the proportion of maximums is $p=0(0.1)1$. Values less than one show RNS performs better than SRS.  }}
  \label{msefixedk-PHR}               
\end{figure}

In Figure \ref{PHR-t123}, we provide the relative precision of the ML estimators of $\lambda$ as the parameter of the exponential distribution in  Example \ref{exp-exp} in three RNS incomplete-data Type I, Type II, and Type III.  The relative precision is calculated by the mean square of the RNS-based ML estimate of $\gamma(\theta)$ over  its SRS-based counterpart, so values less than one are desired. The top left panel shows the relative precision of the RNS-based ML estimator of $\lambda$ in the incomplete-data Type I. The relative precision is calculated for the ML estimators of $\lambda=1, 2, 3, 4$ and  $\zeta\in[0, 1]$. It is seen that for the larger values of $\lambda$, the RNS incomplete-data Type I outperforms SRS for any $\zeta\in [0, 1]$.  As $\zeta$ approaches to  1, regardless of the value of the parameter of interest, the performance of RNS incomplete-data Type I becomes better than SRS. The top right panel presents the relative precision of the RNS incomplete-data Type II for the range of $\lambda=1(0.1)5$ and for four distributions of the set size $K$ as follows
\begin{align*}
\brho_1=(0.4, 0.3, 0.2, 0.1),~\brho_2=(0.1, 0.2, 0.3, 0.4),~\brho_3=(0.2, 0, 0, 0.8),~\brho_4=(0, 0, 0, 1).
\end{align*}
It is seen that the RNS incomplete-data Type II  with the assumed $\brho_1$, $\brho_2$, $\brho_3$, and $\brho_4$ improves the precision of the estimators of $\lambda$ especially when the set sizes get larger. As the value of $\lambda$ increases, the performance of the RNS incomplete-data Type II is improved more and the distributions of $K$ perform similarly. The next four panels in Figure \ref{PHR-t123} present the relative precision of the RNS incomplete-data Type III for $\lambda=1(0.1)5$ and $\zeta\in\{0, 0.25, 0.75, 1\}$. The relative precision for small $\lambda$ depends on $\zeta$. The last four panels in Figure \ref{PHR-t123} show that, for all the considered distributions on $K$,  by increasing $\zeta$, RNS outperforms SRS and the relative precision reaches the lowest value when $\zeta=1$.

\begin{figure}
 \centering  \vspace{-1cm}              
  \includegraphics[width=6in, height=7in]{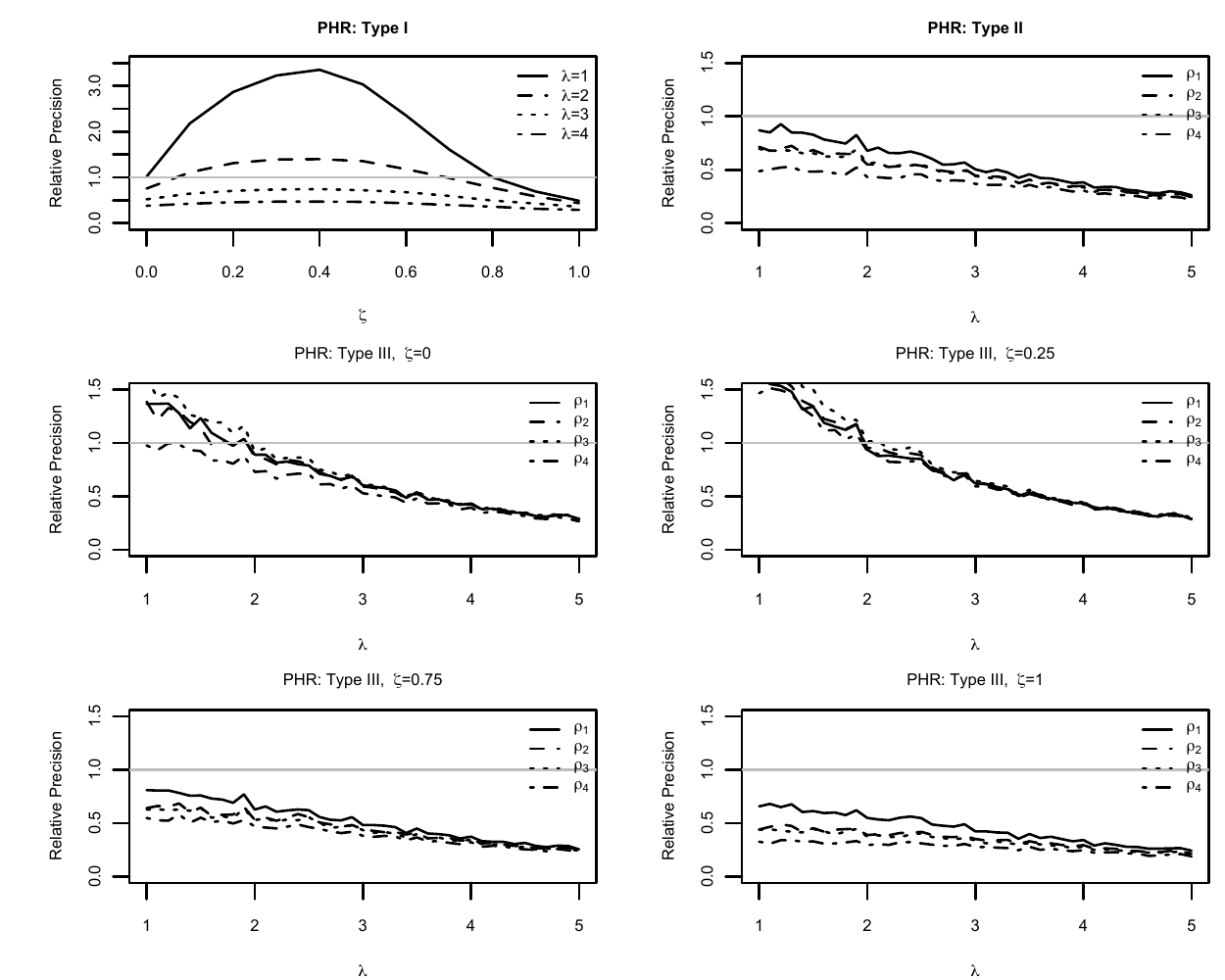} 
  \caption{{The relative precision of the RNS incomplete-data Type I (top left panel) for $\zeta\in[0, 1]$ and $\lambda\in\{1, 2, 3, 4\}$, Type II (top right panel) for four distributions on $K$ and $\lambda=1(0.1)5$, and Type III (middle and lower panels) for $\zeta\in\{0, 0.25, 0.75, 1\}$ and $\lambda=1(0.1)5$ in an exponential distribution with parameter $\lambda$ and $m=10$. Values of the relative precision less than one show RNS performs better than SRS.}}
  \label{PHR-t123}               
\end{figure}

In Figure \ref{PRHR-t123} we present the relative precision of the ML estimators of $\eta$ as the parameter of the beta distribution in the form of Example \ref{beta-beta} for the RNS incomplete-data Types I, II and III. The top left panel shows the relative precision of the RNS-based ML estimator of $\eta=1, 2, 3$ and 4 in the incomplete-data Type I for $\zeta\in[0, 1]$.  It is seen that for  the examined values of $\eta$, $\zeta=0$ improves the RNS incomplete-data Type I over SRS.  The top right panel presents the relative precision of the RNS incomplete-data Type II for the range of $\eta=1(0.1)5$ and for four distributions on $K$, which are shown by $\brho_1$, $\brho_2$, $\brho_3$, and $\brho_4$. For  the examined $\brho$'s, RNS incomplete-data Type II outperforms SRS. The next four panels in Figure \ref{PRHR-t123} present the relative precision of the RNS incomplete-data Type III for $\eta=1(0.1)5$ and $\zeta\in\{0, 0.25, 0.75, 1\}$. It is seen that for $\zeta=0$ the RNS incomplete-data Type III performs better  than SRS. The relative precision of the estimators obtained from the RNS design with $\zeta$ other than zero might works good for some values of $\eta$, especially when $\zeta$ is close to zero.

\begin{figure}
 \centering  \vspace{-1cm}              
  \includegraphics[width=6in, height=7in]{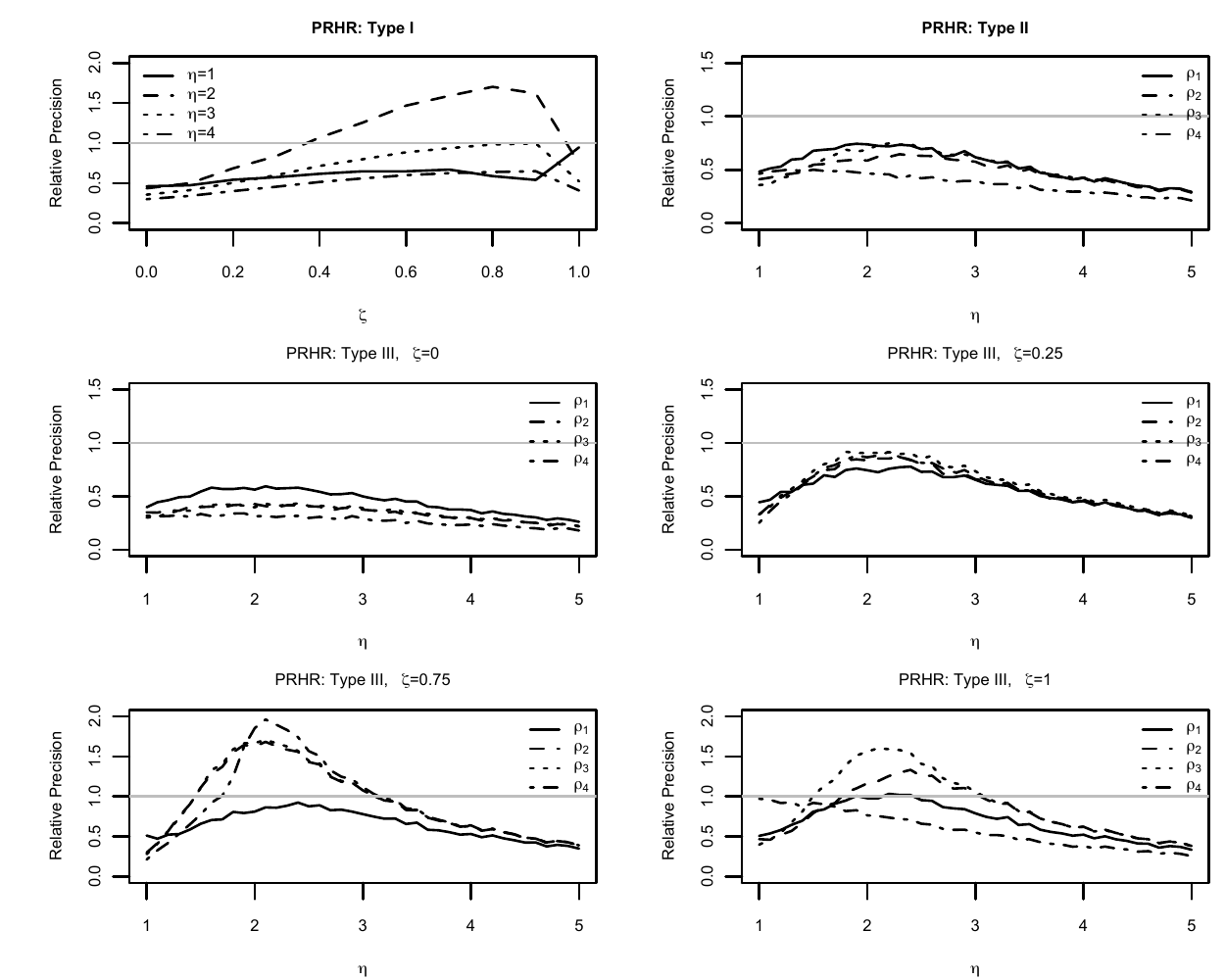} 
 
  \caption{{The relative precision of the RNS incomplete-data Type I (top left panel) for $\zeta\in[0, 1]$ and $\eta\in\{1, 2, 3, 4\}$, Type II (top right panel) for four distributions on $K$ and $\eta=1(0.1)5$, and Type III (middle and lower panels) for $\zeta\in\{0, 0.25, 0.75, 1\}$ and $\eta=1(0.1)5$ in an exponential distribution with parameter $\eta$ and $m=10$. Values of the relative precision less than one shows RNS performs better than SRS.}}
  \label{PRHR-t123}               
\end{figure}

We also evaluated the performance of the RNS-based MM estimators of $\lambda$ and $\eta$. Figure \ref{MME-c} shows the precision of $\widehat{\gamma_{cm}(\bt)}$, the MM estimators of $\gamma(\bt)$  in the RNS setting when the data is complete,  relative to  their SRS counterparts  for the PHR (left panel) and PRHR (right panel) models. The relative precision is calculated  for four RNS designs with fixed set sizes when the sets are of sizes $K=2, 3, 4, 5$ and the proportion of maximums varies in $p=0(0.1)1$. The results show that the RNS design outperforms SRS for all considered distributions of $K$ ($\brho_1$, $\brho_2$, $\brho_3$, and $\brho_4$) and for all proportions of maximums $p\in[0, 1]$.  We observe in the left panel that, similar to the ML estimators in the RNS-based complete-data, increasing the set size and the proportion of maximums improve the performance of the RNS complete-data in the PHR model. In the PHR model, the best performance is obtained from the MANS design, where all the selected units are maximums, and with the set size $K=5$. In the PRHR model  (right panel), the best performance belongs to the MINS design, where all the selected units are minimums,  with the set size $K=5$.

\begin{figure}[t]
 \centering  \vspace{-1cm}              
  \includegraphics[width=6in, height=2.5in]{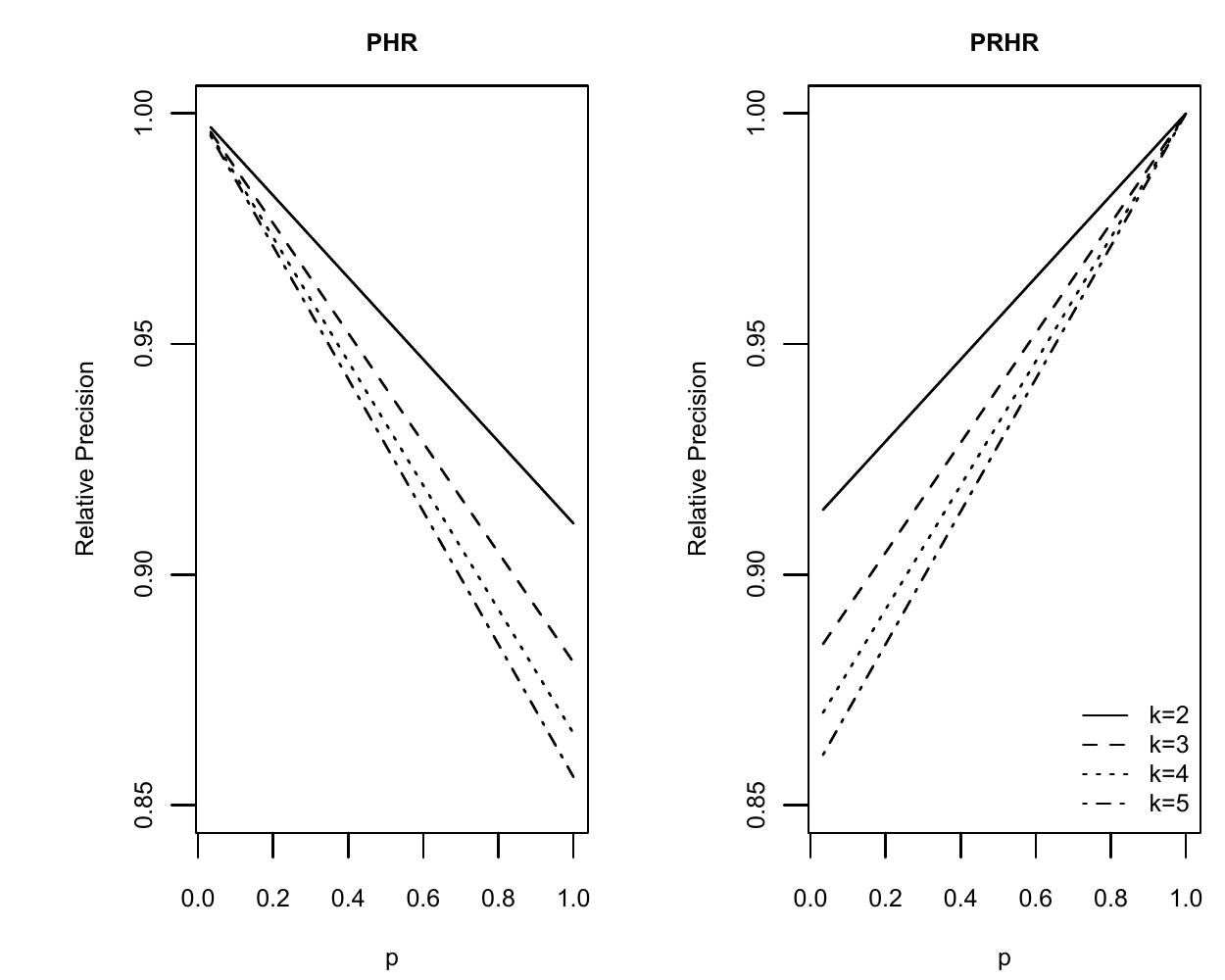} 
  \caption{{The relative precision of the MM estimators of $\gamma(\bt)$  based on the RNS complete-data over  their SRS counterparts in the PHR (left panel) and PRHR (right panel) models when $k_i=2, 3, 4, $ and 5, and the proportion of maximums is $p=0.1(0.1)1$. The relative precision less than one shows RNS performs better than SRS.}}
  \label{MME-c}               
\end{figure}

Figure \ref{MME-t123} provides the relative precision of the MM estimators of parameter $\lambda$ in the exponential distribution introduced in Example \ref{exp-exp} based on the RNS incomplete-data Type I, Type II, and Type III.  The top left panel shows the relative precision of the RNS-based MM estimators in the incomplete-data Type I.  It shows that $\zeta=1$, regardless of the parameter value $\lambda$, is the optimum value of $\zeta$ which improves the RNS-based MM estimator in the incomplete-data Type I scenario compared with SRS.  Looking at the top right panel, which presents the performance of the MM estimator of $\lambda$ in the RNS incomplete-data Type II,  shows that  RNS designs with design parameters $\brho_1$, $\brho_2$, and $\brho_3$ do not perform better than SRS for all the examined parameters $\lambda$.  For $\brho_4$, the performance of the RNS incomplete-data Type II improves over SRS. The next  four panels in the middle and down in Figure \ref{MME-t123} present the performance of the MM estimators of $\lambda$ in the RNS incomplete-data Type II for $\zeta\in\{0, 0.25, 0.75, 1\}$, where $\zeta=1$ has the best impact on the performance of this type of the RNS design especially when the parameter value $\lambda$ increases. For small $\lambda$ the proportion of SRS, samples from the sets of size $K=1$, should be small.

\begin{figure}
 \centering  \vspace{-1cm}              
  \includegraphics[width=6in, height=7in]{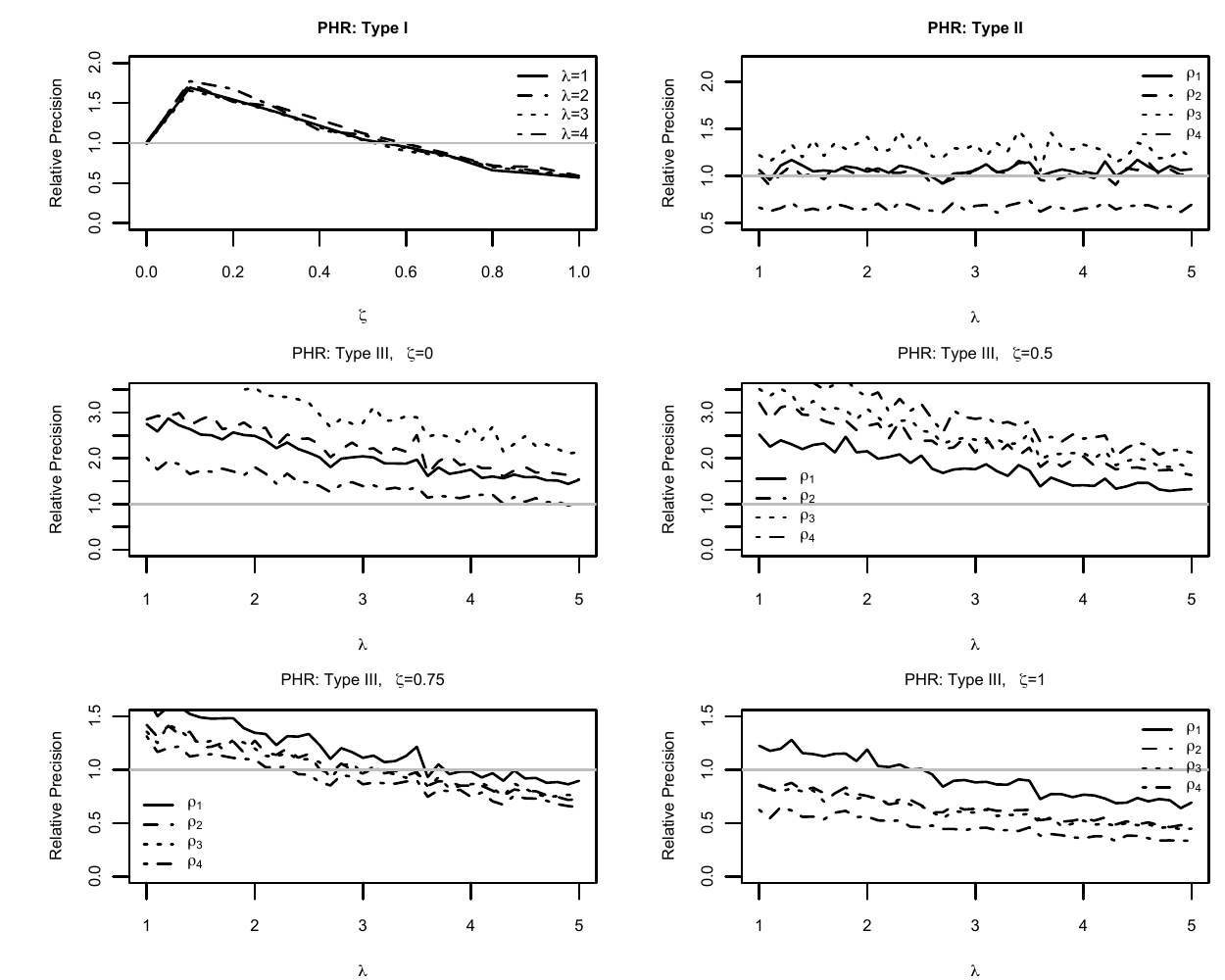} 
 \caption{{The relative precision of the RNS incomplete-data Type I (top left panel) for $\zeta\in[0, 1]$ and $\lambda\in\{1, 2, 3, 4\}$, Type II (top right panel) for four distributions on $K$ and $\eta=1(0.1)5$, and Type III (middle and lower panels) for $\zeta\in\{0, 0.25, 0.75, 1\}$ and $\eta=1(0.1)5$ in a beta distribution with parameter $1/\eta$, the shape parameter $\beta=1$, and $m=10$. Values of the relative precision less than one shows RNS performs better than SRS.}}
  \label{MME-t123}               
\end{figure}

\begin{figure}
 \centering  \vspace{-1cm}              
  \includegraphics[width=6in, height=7in]{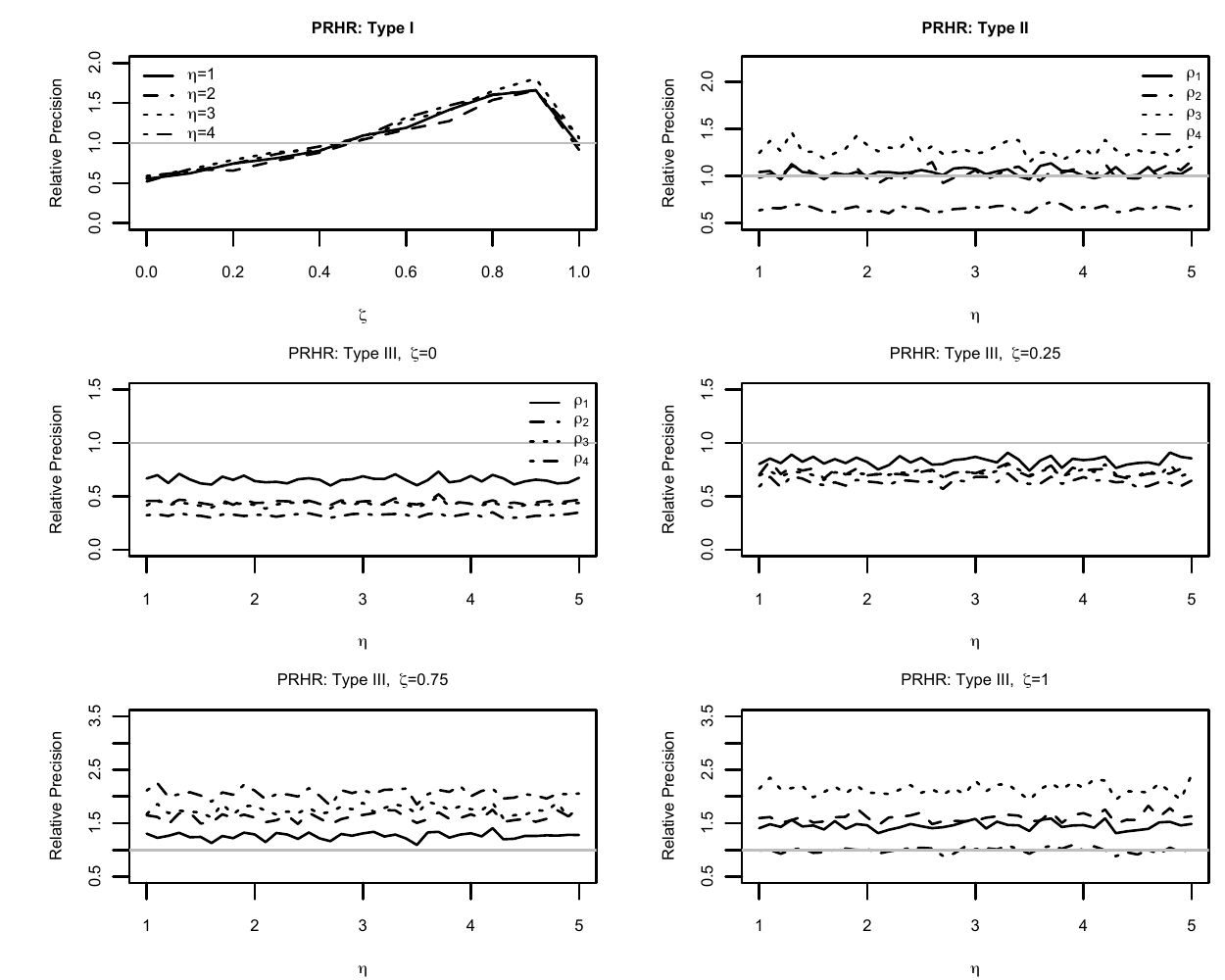} 
 
  \caption{{The relative precision of the RNS incomplete-data Type I (top left panel) for $\zeta\in[0, 1]$ and $\eta\in\{1, 2, 3, 4\}$, Type II (top right panel) for four distributions on $K$ and $\eta=1(0.1)5$, and Type III (middle and lower panels) for $\zeta\in\{0, 0.25, 0.75, 1\}$ and $\eta=1(0.1)5$ in a beta distribution with parameter $1/\eta$, the shape parameter $\beta=1$, and $m=10$. Values of the relative precision less than one shows RNS performs better than SRS. }}
  \label{PRHR-MM-t123}               
\end{figure}

In Figure \ref{PRHR-MM-t123}, the performance of the RNS-based MM estimators of $\eta$ in the introduced beta distribution are compared with their corresponding estimators in the SRS design. To evaluate the relative precision in the RNS incomplete-data Type I, the RNS-based MM estimators of $\eta=1, 2, 3, 4$ for $\zeta\in[0, 1]$ are examined.   The top left panel shows that, no matter what the value of $\eta$ is,  $\zeta=0$ provides the best performance of the RNS incomplete-data Type I compared with the SRS scheme.  Considering the top right panel, which shows the relative precision of the RNS-based MM estimators in the incomplete-data Type II, it is seen that for all assumed distributions on $K$ with larger and fixed set sizes, i.e., $\brho_4$ and regardless of the parameter value $\eta$, the RNS outperforms SRS. The next panels in Figure \ref{PRHR-MM-t123} confirm that $\zeta=0$ provides the most efficient RNS-based MM estimators of $\eta$ in the incomplete-data Type III. As $\zeta$ increases, for all values of $\eta$ the performance of this type of RNS  gets worse.

\subsection{A Case Study}

\begin{figure}
 \centering  \vspace{-1cm}              
  \includegraphics[width=6in, height=2.75in]{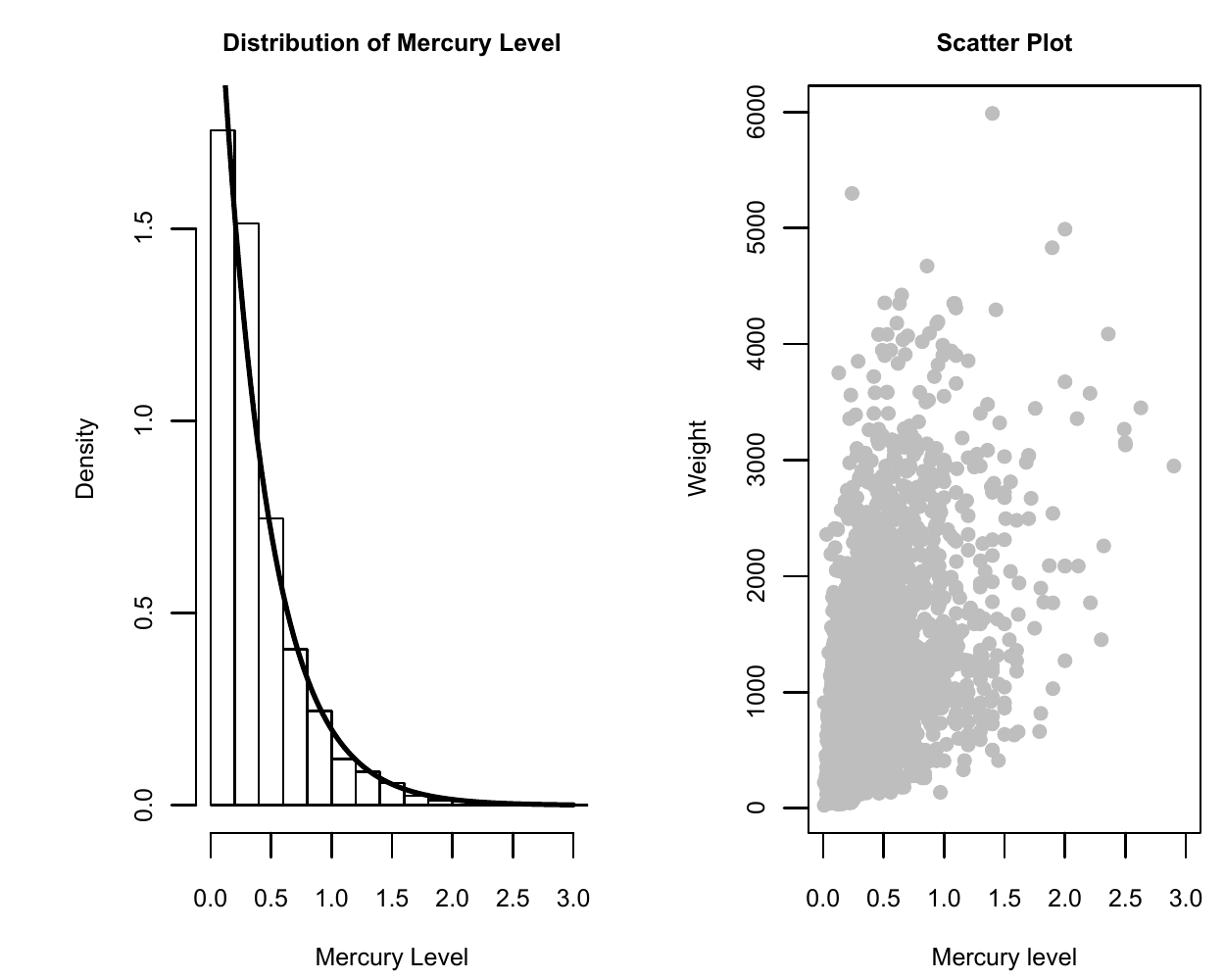} 
 
  \caption{The population shape of the fish mercury level (left panel) and scatterplot between the fish mercury level and fish weight.}
  \label{RD-fig1}               
\end{figure}

Fish products have been shown to contain varying amounts of heavy metals, particularly mercury  from water pollution.  
Mercury is dangerous to both natural ecosystems and humans because it is a metal known to be highly toxic, especially due to its ability to damage the central nervous system. Children, as well as women who are pregnant or planning to become pregnant, are the most vulnerable to mercury's harmful effects. Many studies have been performed to determine the mercury concentration in the fish species and evaluate the performance of the proposed remedial actions (e.g. \citet{bhavsar2010changes}, \citet{mcgoldrick2010canada} and  \citet{nourmohammadi2015distribution}). 

\begin{figure}[t]
 \centering                
  \includegraphics[width=6.5in, height=2.5in]{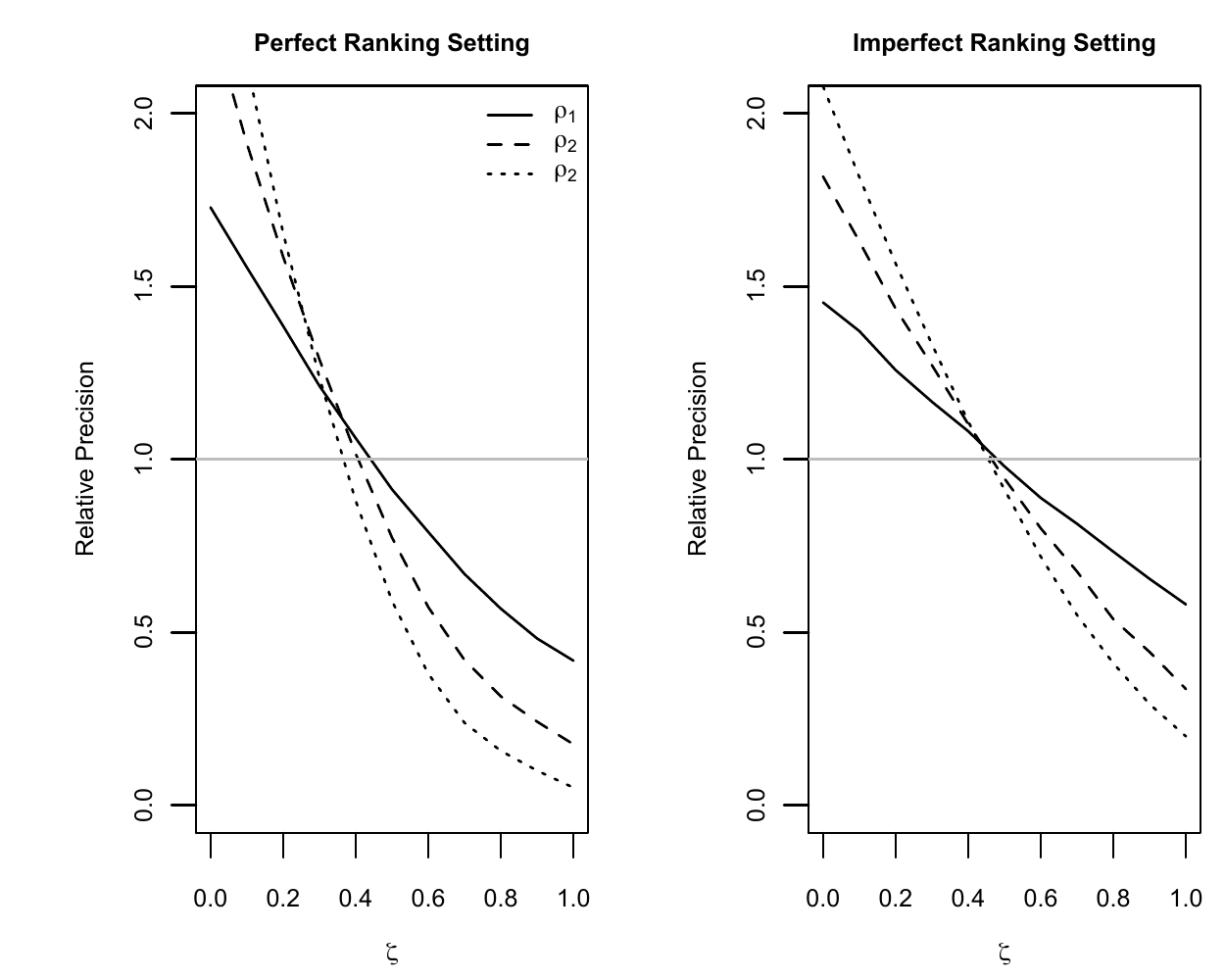} 
 
  \caption{The Relative precision of the ML estimators of $\lambda$ based on RNS incomplete-data Type III over their SRS counterparts in the perfect ranking setting (left panel) and imperfect ranking setting (right panel) when  the parameters $\brho$ and $\zeta$ are assumed to be unknown for $\zeta=0(0.1)1$ and $\brho_1$, $\brho_2$, $\brho_3$.}
  \label{RD-fig2}               
\end{figure}

The mercury grows in concentration within the bodies of fish. It is well-known that measuring the mercury level in fish needs  an expensive and complex laboratory procedure. However, a small group of fish can be ranked based on their    mercury contamination levels before taking the actual measurements on them. This can be done using either  the fish  weight or length which have  positive correlations with mercury contamination levels in  fish. In this section, the performance of the RNS-based modified maximum likelihood estimator of the distribution parameter is compared with its corresponding SRS counterpart. The dataset contains mercury levels, weights, and lengths of 3033 {\it sandra vitreus} (Walleye) caught in Minnesota. The original database contains 102,850 fish tissue mercury records compiled by the United States Geological Survey from various sources. Selected records were those that: had a non-zero mercury concentration; had a valid species identifier; had non-zero length and weight measurements; and had a sample type of ``filet skin on".

In this study both perfect and imperfect ranking settings are considered. In the perfect ranking setting the study variable, i.e. the fish mercury level, is used for ranking the sample units in the sets and for the imperfect ranking setting, ranking is performed using the fish weights. Kendall's $\tau$ between the mercury level and weight values is about 0.4, which is relatively small.  Figure \ref{RD-fig1} (left panel) shows the population shape  for the mercury levels of 3033 fish records and it looks  the fish mercury level follows an exponential distribution, a member of PHR model,  with parameter $\lambda=0.3902$. The choice of exponential distribution for fish mercury level is justified by the Kolmogorov-Smirnov test.  Figure \ref{RD-fig1} (right panel) shows the scatterplot between the mercury level and weight. The RNS complete-data and three RNS incomplete-data scenarios were considered in section \ref{simulation}. In the examined three types RNS incomplete-data, the parameters $\rho$ and/or $\zeta$ are assumed to be known. In this section we examine the performance of RNS incomplete-data Type III in which the parameters $\rho$ and $\zeta$ are unknown. The performance of the design in estimating the population parameter is evaluated using the relative precision, i.e. the mean square of RNS-based MLE of $\lambda$ over the mean square of its SRS counterpart. The relative precision is calculated for the MLE's of $\lambda$  for $\zeta=0(0.1)1$ and three distributions of the set size $K$ as follow
\begin{align*}
\brho_1=(0.4, 0.3, 0.2, 0.1),~\brho_2=(0.1, 0.2, 0.3, 0.4),~\brho_3=(0, 0, 0, 1).
\end{align*}
Figure \ref{RD-fig2} provides the relative precision of the ML estimators of parameter $\lambda$  based on the RNS incomplete-data Type III. It shows that $\zeta=1$, regardless of the parameter value $\rho$, is the optimum value of $\zeta$ which improves the RNS-based ML estimator in the incomplete-data Type II scenario compared with their SRS counterparts. The relative precision of RNS-based estimators presented in Figure \ref{RD-fig2} is obtained in an EM algorithm when the initial values of $\brho$ and $\zeta$ are $\brho_0=(0.25, 0.25, 0.25, 0.25)$ and $\zeta_0=1$. Considering the sensitivity of the EM algorithm to the initial values of the unknown parameters $\brho$ and $\zeta$, we also examined $\brho_0=(1, 0, 0, 0)$ , $\brho_0=(0, 0, 0, 1)$, $\zeta_0=0$ and $\zeta_0=0.5$. Except $\brho_0=(1, 0, 0, 0)$ and $\zeta_0=0$, for the other  examined initial values of $\brho$ and $\zeta$, RNS outperforms SRS for larger values of true $\zeta$'s, i.e. $\zeta\in(0.5, 1]$, and it shows the best performance of RNS over SRS at $\zeta=1$.

\section{Concluding remarks}\label{conclusion}

Randomized nomination sampling (RNS) was introduced by  \citet{jafari2012randomized} and it has been shown to perform better than simple random sampling (SRS) in constructing nonparametric confidence and tolerance intervals. RNS has potentials for a wide range of applications in medical, environmental and ecological studies. 
In this paper, we described the RNS-based ML and MM estimators of the population parameters when the underlying study variable follows PHR or PRHR model. Various conditions on the type of information, ranking error settings and the design parameters including distribution of the set size ($\brho$) and  probability of taking the maximum observation of the set ($\zeta$) have been investigated. We introduced four types of RNS data, corresponding to situations in which all the observations, set sizes and observations ranks in the sets are known, only observations and the set sizes are known, only the observations and their ranks in the sets are known, or finally only the observations are known.  Considering all the situations, we also provided the PDF and CDF of an RNS  observation.  We showed that there is always a range of $\zeta$ on each RNS is superior to SRS in terms of the relative precision. The RNS design has this advantage regardless of the ranking setting.  The relative precision of the estimators obtained in the RNS design becomes better when more weight is given to the larger set size and $\zeta=1$ ( in PHR model) or $\zeta=0$ (in PRHR model).



\section*{Acknowledgements}

The authors gratefully  acknowledge the partial support of  the NSERC Canada.


\bibliographystyle{chicago}

\begin{thebibliography}{}

\bibitem[\protect\citeauthoryear{Al-Odat and Al-Saleh}{Al-Odat and
  Al-Saleh}{2001}]{al2001variation}
Al-Odat, M. and M.~F. Al-Saleh (2001).
\newblock A variation of ranked set sampling.
\newblock {\em Journal of Applied Statistical Science\/}~{\em 10\/}(2),
  137--146.

\bibitem[\protect\citeauthoryear{Bhavsar, Gewurtz, McGoldrick, Keir, and
  Backus}{Bhavsar et~al.}{2010}]{bhavsar2010changes}
Bhavsar, S.~P., S.~B. Gewurtz, D.~J. McGoldrick, M.~J. Keir, and S.~M. Backus
  (2010).
\newblock Changes in mercury levels in great lakes fish between 1970s and 2007.
\newblock {\em Environmental science \& technology\/}~{\em 44\/}(9),
  3273--3279.

\bibitem[\protect\citeauthoryear{Boyles and Samaniego}{Boyles and
  Samaniego}{1986}]{boyles1986estimating}
Boyles, R.~A. and F.~J. Samaniego (1986).
\newblock Estimating a distribution function based on nomination sampling.
\newblock {\em Journal of the American Statistical Association\/}~{\em
  81\/}(396), 1039--1045.

\bibitem[\protect\citeauthoryear{Gemayel, Stasny, and Wolfe}{Gemayel
  et~al.}{2010}]{gemayel2010optimal}
Gemayel, N.~M., E.~A. Stasny, and D.~A. Wolfe (2010).
\newblock Optimal ranked set sampling estimation based on medians from multiple
  set sizes.
\newblock {\em Journal of Nonparametric Statistics\/}~{\em 22\/}(4), 517--527.

\bibitem[\protect\citeauthoryear{Ghosh and Tiwari}{Ghosh and
  Tiwari}{2009}]{ghosh2009unified}
Ghosh, K. and R.~C. Tiwari (2009).
\newblock A unified approach to variations of ranked set sampling with
  applications.
\newblock {\em Journal of Nonparametric Statistics\/}~{\em 21\/}(4), 471--485.

\bibitem[\protect\citeauthoryear{Gupta, Gupta, and Gupta}{Gupta
  et~al.}{1998}]{gupta1998modeling}
Gupta, R.~C., P.~L. Gupta, and R.~D. Gupta (1998).
\newblock Modeling failure time data by lehman alternatives.
\newblock {\em Communications in Statistics-Theory and Methods\/}~{\em
  27\/}(4), 887--904.

\bibitem[\protect\citeauthoryear{Gupta and Gupta}{Gupta and
  Gupta}{2007}]{gupta2007proportional}
Gupta, R.~C. and R.~D. Gupta (2007).
\newblock Proportional reversed hazard rate model and its applications.
\newblock {\em Journal of Statistical Planning and Inference\/}~{\em
  137\/}(11), 3525--3536.

\bibitem[\protect\citeauthoryear{Helsen and Schmittlein}{Helsen and
  Schmittlein}{1993}]{helsen1993analyzing}
Helsen, K. and D.~C. Schmittlein (1993).
\newblock Analyzing duration times in marketing: Evidence for the effectiveness
  of hazard rate models.
\newblock {\em Marketing Science\/}~{\em 12\/}(4), 395--414.

\bibitem[\protect\citeauthoryear{Jafari~Jozani and Johnson}{Jafari~Jozani and
  Johnson}{2012}]{jafari2012randomized}
Jafari~Jozani, M. and B.~C. Johnson (2012).
\newblock Randomized nomination sampling for finite populations.
\newblock {\em Journal of Statistical Planning and Inference\/}~{\em 142\/}(7),
  2103--2115.

\bibitem[\protect\citeauthoryear{Jafari~Jozani and Mirkamali}{Jafari~Jozani and
  Mirkamali}{2010}]{jafari2010improved}
Jafari~Jozani, M. and S.~J. Mirkamali (2010).
\newblock Improved attribute acceptance sampling plans based on maxima
  nomination sampling.
\newblock {\em Journal of Statistical Planning and Inference\/}~{\em 140\/}(9),
  2448--2460.

\bibitem[\protect\citeauthoryear{Jafari~Jozani and Mirkamali}{Jafari~Jozani and
  Mirkamali}{2011}]{jafari2011control}
Jafari~Jozani, M. and S.~J. Mirkamali (2011).
\newblock Control charts for attributes with maxima nominated samples.
\newblock {\em Journal of Statistical Planning and Inference\/}~{\em 141\/}(7),
  2386--2398.

\bibitem[\protect\citeauthoryear{Kvam and Samaniego}{Kvam and
  Samaniego}{1993}]{kvam1993estimating}
Kvam, P.~H. and F.~J. Samaniego (1993).
\newblock On estimating distribution functions using nomination samples.
\newblock {\em Journal of the American Statistical Association\/}~{\em
  88\/}(424), 1317--1322.

\bibitem[\protect\citeauthoryear{Lawless}{Lawless}{2011}]{lawless2011statistical}
Lawless, J.~F. (2011).
\newblock {\em Statistical models and methods for lifetime data}, Volume 362.
\newblock John Wiley \& Sons.

\bibitem[\protect\citeauthoryear{McGoldrick, Clark, Keir, Backus, and
  Malecki}{McGoldrick et~al.}{2010}]{mcgoldrick2010canada}
McGoldrick, D.~J., M.~G. Clark, M.~J. Keir, S.~M. Backus, and M.~M. Malecki
  (2010).
\newblock Canada's national aquatic biological specimen bank and database.
\newblock {\em Journal of Great Lakes Research\/}~{\em 36\/}(2), 393--398.

\bibitem[\protect\citeauthoryear{Mehrotra and Nanda}{Mehrotra and
  Nanda}{1974}]{mehrotra1974unbiased}
Mehrotra, K. and P.~Nanda (1974).
\newblock Unbiased estimation of parameters by order statistics in the case of
  censored samples.
\newblock {\em Biometrika\/}~{\em 61\/}(3), 601--606.

\bibitem[\protect\citeauthoryear{Navarro, Samaniego, Balakrishnan, and
  Bhattacharya}{Navarro et~al.}{2008}]{navarro2008application}
Navarro, J., F.~J. Samaniego, N.~Balakrishnan, and D.~Bhattacharya (2008).
\newblock On the application and extension of system signatures in engineering
  reliability.
\newblock {\em Naval Research Logistics (NRL)\/}~{\em 55\/}(4), 313--327.

\bibitem[\protect\citeauthoryear{Nourmohammadi, Jafari~Jozani, and
  Johnson}{Nourmohammadi et~al.}{2014}]{nourmohammadi2014confidence}
Nourmohammadi, M., M.~Jafari~Jozani, and B.~C. Johnson (2014).
\newblock Confidence intervals for quantiles in finite populations with
  randomized nomination sampling.
\newblock {\em Computational Statistics \& Data Analysis\/}~{\em 73}, 112--128.

\bibitem[\protect\citeauthoryear{Nourmohammadi, Jafari~Jozani, and
  Johnson}{Nourmohammadi et~al.}{2015a}]{nourmohammadi2015distribution}
Nourmohammadi, M., M.~Jafari~Jozani, and B.~C. Johnson (2015a).
\newblock Distribution-free tolerance intervals with nomination samples:
  Applications to mercury contamination in fish.
\newblock {\em Statistical Methodology\/}~{\em 26}, 16--33.

\bibitem[\protect\citeauthoryear{Nourmohammadi, Jafari~Jozani, and
  Johnson}{Nourmohammadi et~al.}{2015b}]{nourmohammadi2014nonparametric}
Nourmohammadi, M., M.~Jafari~Jozani, and B.~C. Johnson (2015b).
\newblock Nonparametric confidence intervals for quantiles with randomized
  nomination sampling.
\newblock {\em Sankhya A\/}~{\em 77\/}(2), 408--432.

\bibitem[\protect\citeauthoryear{Samawi, Ahmed, and Abu-Dayyeh}{Samawi
  et~al.}{1996}]{samawi1996estimating}
Samawi, H.~M., M.~S. Ahmed, and W.~Abu-Dayyeh (1996).
\newblock Estimating the population mean using extreme ranked set sampling.
\newblock {\em Biometrical Journal\/}~{\em 38\/}(5), 577--586.

\bibitem[\protect\citeauthoryear{Tiwari}{Tiwari}{1988}]{tiwari1988nonparametric}
Tiwari, R.~C. (1988).
\newblock Nonparametric bayes estimation of a distribution under nomination
  sampling.
\newblock {\em Reliability, IEEE Transactions on\/}~{\em 37\/}(5), 558--561.

\bibitem[\protect\citeauthoryear{Tiwari and Wells}{Tiwari and
  Wells}{1989}]{tiwari1989quantile}
Tiwari, R.~C. and M.~T. Wells (1989).
\newblock Quantile estimation based on nomination sampling.
\newblock {\em Reliability, IEEE Transactions on\/}~{\em 38\/}(5), 612--614.

\bibitem[\protect\citeauthoryear{Wells, Tiwari, et~al.}{Wells
  et~al.}{1990}]{wells1990estimating}
Wells, M.~T., R.~C. Tiwari, et~al. (1990).
\newblock Estimating a distribution function based on minima-nomination
  sampling.
\newblock In {\em Topics in statistical dependence}, pp.\  471--479. Institute
  of Mathematical Statistics.

\bibitem[\protect\citeauthoryear{Willemain}{Willemain}{1980}]{willemain1980estimating}
Willemain, T.~R. (1980).
\newblock Estimating the population median by nomination sampling.
\newblock {\em Journal of the American Statistical Association\/}~{\em
  75\/}(372), 908--911.

\end{thebibliography}

\end{document}